\documentclass[journal]{IEEEtran}
\usepackage{amsmath}
\usepackage{amsfonts}
\usepackage{amsthm}
\usepackage{cite}
\usepackage{soul,color}
\usepackage{graphicx}
\usepackage{float}
\usepackage[font=footnotesize]{caption}
\usepackage[font=footnotesize]{subcaption}
\usepackage{epstopdf}
\usepackage{amssymb}
\usepackage{stmaryrd}
\usepackage{multirow}
\newtheorem{theorem}{\textit{Theorem}}
 
\usepackage{enumitem}
\usepackage{multirow}
\usepackage{algorithm}
\usepackage{algorithmic}
\DeclareMathOperator\erf{erf}
\begin{document} 
\title{CoMP-Assisted NOMA and Cooperative NOMA in Indoor VLC Cellular Systems}
\author{Mohamed Amine Arfaoui$^{*}$,
        Ali Ghrayeb, 
        Chadi Assi, and
        Marwa Qaraqe\vspace{-0.7cm}
\thanks{M. A. Arfaoui and C. Assi are with Concordia Institute for Information Systems Engineering (CIISE), Concordia University, Montreal, Quebec, Canada, e-mails:\{m\_arfaou@encs, assi@ciise\}.concordia.ca.}
\thanks{A. Ghrayeb is with the Electrical and Computer Engineering (ECE) department, Texas A$\&$M University at Qatar, Doha, Qatar, e-mail: ali.ghrayeb@qatar.tamu.edu.}
\thanks{Marwa Qaraqe is with College of Science and Engineering, Hamad Bin Khalifa University, Qatar Foundation, Doha, Qatar, e-mail: mqaraqe@hbku.edu.qa.} 
\thanks{$^*$\textit{Corresponding author: M. A. Arfaoui, m\_arfaou@.concordia.ca}}
}
\maketitle
\thispagestyle{plain}
\begin{abstract}
In this paper, we investigate the dynamic power allocation for a visible light communication (VLC) cellular system consisting of two coordinating attocells, each equipped with one access-point (AP). The coordinated multipoint (CoMP) between the two cells is introduced to assist users experiencing high inter-cell-interference (ICI). Specifically, the coordinated zero-forcing (ZF) precoding is used to cancel the ICI at the users located near the centers of the cells, whereas the joint transmission (JT) is employed to eliminate the ICI at the users located at the edge of both cells and to improve their receptions as well. Furthermore, two multiple access techniques are invoked within each cell, namely, non-orthogonal-multiple-access (NOMA) and cooperative non-orthogonal-multiple-access (C-NOMA). Hence, two multiple access techniques are proposed for the considered multi-user multi-cell system, namely, the CoMP-assisted NOMA scheme and the CoMP-assisted C-NOMA scheme. For each scheme, two power allocation frameworks are formulated each as an optimization problem, where the objective of the former is maximizing the network sum data rate while guaranteeing a certain quality-of-service (QoS) for each user, whereas the goal of the latter is to maximize the minimum data rate among all coexisting users. The formulated optimization problems are not convex, and hence, difficult to be solved directly unless using heuristic methods, which comes at the expense of high computational complexity. To overcome this issue, optimal and low complexity power allocation schemes are derived. In the simulation results, the performance of the proposed CoMP-assisted NOMA and CoMP-assisted C-NOMA schemes are compared with those of the CoMP-assisted orthogonal-multiple-access (OMA) scheme, the C-NOMA scheme and the NOMA scheme, where the superiority of the proposed schemes are demonstrated. Finally, the performance of the proposed schemes and the considered baselines is evaluated while varying various system parameters.
\end{abstract} 
\IEEEpeerreviewmaketitle
\section{Introduction}
\subsection{Motivation}
\indent As the fifth generation (5G) wireless networks are currently under deployment, researchers from both academia and industry started shaping their vision on how the upcoming sixth generation (6G) should be \cite{saad2019vision}. The main goals of 6G networks are not only to fill the gap of the original and unfulfilled promises of 5G or to keep up with the continuous emergence of the Internet of-Things (IoTs) networks but also to handle the exponential increase of the number of devices connected to the Internet, which is predicted to reach $29.3$ billion networked devices by $2023$ \cite{cisco2020cisco,zhang20196g}. Therefore, 6G networks must urgently provide high data rates, seamless and massive connectivity, ubiquitous coverage and ultra-low latency communications in order to reach the preset targets \cite{zhang20196g}. Due to this, researchers from industry and academia are trying to explore new network architectures, new transmission techniques and higher frequency bands, such as the millimeter wave (mmWave), the terahertz (THz), the infrared, and the visible light bands, to meet these high demands \cite{zhang20196g}. \\
\indent Visible light communication (VLC) is an emerging high speed optical wireless communication technology that uses the visible light as the propagation medium in the downlink for the purposes of illumination and wireless communication. VLC operates in the visible light spectrum to communicate data between access points (APs) and users through the existing illumination components, such as the Light Emitting Diodes (LEDs). VLC offers a number of important benefits that have made it favorable for 6G networks \cite{david20186g}, such as the large unregulated visible light band, which translates into higher data rates and higher connectivity in comparison to traditional radio-frequency (RF) networks \cite{haas2015lifi}, the high energy efficiency \cite{tavakkolnia2018energy}, the straightforward deployment that uses off-the-shelf LEDs and photodiodes (PDs) devices at the transmitter and the receiver ends, respectively, and the enhanced security since light does not penetrate through opaque objects \cite{arfaoui2020physical}. \\
\indent As a wireless broadband technology, VLC must support a high number of users with simultaneous network access \cite{obeed2019optimizing}. Traditional multiple access techniques, referred to as orthogonal multiple access (OMA) techniques, allocate the available resources to coexisting users in an orthogonal manner in order to cancel the inter-user interference (IUI). Such multiple access techniques include time division multiple access (TDMA), and frequency division multiple access (FDMA) \cite{obeed2019optimizing}. The main drawback of the aforementioned techniques is that the maximum number of users that can be served is limited by the number of available orthogonal resources. In other words, the OMA techniques cannot provide sufficient resource reuse when the number of coexisting users is high. This makes OMA techniques unable to support a massive number of users, and hence, unable to provide a massive connectivity, which is one of the main requirements of 6G networks. \\
\indent In an effort to increase the throughput and improve the fairness of VLC systems, the non-orthogonal multiple access (NOMA) technique was introduced in the literature. In contrast to OMA techniques, NOMA allows multiple users to exploit the same time/frequency resource blocks at the expense of some IUI, leading to an efficient resource utilization. In particular, downlink power-domain NOMA relies on the superposition coding (SC) concept at the transmitter side to multiplex the data streams of different users in the power domain, and on the successive interference cancellation (SIC) concept at the end users to decode their received data \cite{zhang2014performance}.\footnote{In this paper, the term “NOMA” is restricted to power-domain NOMA as distinct from its code-domain NOMA counterpart.} NOMA operates by allocating different power levels to users based on their channel gains. To decode the data, the strong users first apply SIC to decode the weak users' data, cancel it from their respective receptions, and then decode their data, whereas the weak users proceed to decode their data directly, and hence, suffer from IUI resulting from the superposition of the strong users' data. \\
\indent Despite the great benefits that VLC offers, it suffers from various shortcomings that make the current technology still far from satisfying the demands of 6G networks. The first limitation is the short communication range resulting from the short wavelengths of the visible light waves. This results in high propagation losses as the VLC channel gain significantly deteriorates when the distance between transmitting and receiving devices increases, in addition to the fact that the visible light can be easily blocked by obstacles \cite{zeng2019angle,soltani2019bidirectional,arfaoui2020measurements}. Moreover, unlike conventional RF wireless systems, the VLC channel is not isotropic, meaning that the orientations of the transmitting and receiving devices affect the channel gains significantly \cite{zeng2019angle,soltani2019bidirectional,arfaoui2020measurements}. As a result, the VLC channel quality fluctuates and the performance of advanced multiple access techniques, such as NOMA, is significantly affected when applied to VLC systems \cite{zeng2019angle,soltani2019bidirectional,arfaoui2020measurements}. 
\subsection{Related Works}
A large body of work has been produced in the application of NOMA in VLC systems, such as \cite{kizilirmak2015non,yin2016performance,yapici2019noma,naser2020rate,obeed2020user} to name a few. In \cite{kizilirmak2015non,yin2016performance}, it was shown that NOMA outperforms the traditional orthogonal frequency division multiple access (OFDMA) scheme in indoor VLC systems that are serving multiple users with fixed positions and orientations. Both works prove the superiority of NOMA over OMA for stationary users but refrain from studying the performance of NOMA for mobile users with random positions and orientations. In \cite{yapici2019noma}, NOMA was investigated in a downlink multi-user VLC system with mobile and randomly oriented users, where the sum-rate and the outage probability were derived. Recently, an overview of the key multiple access techniques used in single-cell VLC systems, such as NOMA, space-division-multiple-access (SDMA), and rate splitting multiple access (RSMA), was provided in \cite{naser2020rate}. In \cite{obeed2020user}, a system model consisting of one AP serving multiple users simultaneously was established, where users are served using the cooperative NOMA (C-NOMA) scheme.  C-NOMA is an enhanced version of NOMA that takes advantage of the desirable attributes of NOMA and device-to-device (D2D) communication. By exploiting the SIC capabilities, each strong user can act as a relay to assist the communication between the transmitter and the weak user through an RF D2D link. Hence, each weak user receives multiple versions of his data, one through the VLC link coming from the transmitter and the remaining from the RF links coming from the associated strong users. In \cite{obeed2020user}, the VLC system sum-rate was maximized by optimizing the strong/weak user pairing, the VLC/RF link selection, and the transmitted power, where it was shown that C-NOMA outperforms the conventional NOMA scheme. The main drawback with the works in \cite{kizilirmak2015non,yin2016performance,yapici2019noma,naser2020rate,obeed2020user} is that the performance of NOMA was investigated in a single-cell setup, and the extension to multi-cell configuration was not considered\\
\indent For 5G wireless networks and beyond, cell densification has been demonstrated to be an effective method to increase the network capacity. The main motivation behind cell densification is reducing the path loss and allowing the reuse of partial or total spectrum by small cells within a given coverage area. In the context of VLC, the concept of optical attocell was first introduced in \cite{haas2013high}. However, similar to the cell densification concept, the main drawback of the use of multiple optical attocells in indoor environments is the severe inter-cell interference (ICI). Precisely, an indoor environment can be composed of multiple optical attocells, each having a radius of around $3$ m \cite{yin2016performance}. These optical attocells are adjacent to each other. Therefore, when the optical attocells are exploiting the same frequency resources, the users within one attocell can experience severe ICI from adjacent cells. \\ 
\indent The application of NOMA in multicell VLC systems was studied in \cite{zhang2016user,obeed2020power}. In \cite{zhang2016user}, a user grouping scheme based on users locations was proposed to reduce the ICI effects in NOMA-based multi-cell VLC networks. With the residual interference from the SIC process in NOMA taken into account, the power allocation within each attocell was optimized to improve the achievable rate per user under a quality-of-service (QoS) constraint. Recently, a multi-cell VLC system was considered in \cite{obeed2020power}, where each attocell consists of an AP that serves two users coexisting within its coverage using C-NOMA. However, the main drawback of the schemes proposed in \cite{zhang2016user,obeed2020power} is that the VLC users still suffer from the ICI effects since no ICI mitigation techniques were employed. In order to mitigate the ICI effects, the concept of cooperative cellular systems has been introduced in practical VLC systems, where multiple VLC APs coordinate together in serving multiple users within the resulting illuminated area \cite{yang2020coordinated,chen2013joint,ma2015coordinated,ma2013integration,yang2018joint,pham2017multi_1}. In this context, the users that are highly prone to ICI effects can be jointly served by a set of adjacent APs. This ICI mitigation technique is referred to as the coordinated multipoint (CoMP) or the coordinated broadcasting technique \cite{ali2018downlink,yang2020coordinated,chen2013joint,ma2015coordinated,ma2013integration,yang2018joint,pham2017multi_1}. \\ 
\indent The performance of CoMP transmission in multi-cell downlink RF networks has been investigated in \cite{ali2018downlink,elhattab2020comp,elhattab2020joint,elhattab2020power,elhattab2022ris,elhattab2022joint}, where different system configurations and various performance metrics were considered. However, different from RF systems, a constraint is imposed on the amplitude of the transmitted VLC signals, which is referred to as the peak-power constraint, rather than on their average power. Due to this, the power allocation schemes developed for multi-user multi-cell RF systems do not apply for VLC systems, which is the case in \cite{obeed2019optimizing,obeed2020power} to name a few. On the other hand, the performance of CoMP transmission in VLC systems was also investigated in the literature \cite{yang2020coordinated,chen2013joint,ma2015coordinated,ma2013integration,yang2018joint,pham2017multi_1}. The authors in \cite{yang2020coordinated} considered the problem of joint resource and power allocation in OFDMA-based coordinated multi-cell network. OFDMA is an OMA technique and it was demonstrated in the literature that NOMA techniques offer a higher spectral efficiency and a greater connectivity when compared to OMA techniques. In \cite{chen2013joint}, it was demonstrated that the CoMP technique can achieve higher signal-to-interference-plus-noise ratios (SINRs) for ICI prone users in comparison to the frequency reuse (FR) technique. In \cite{ma2015coordinated,ma2013integration}, two linear precoders based on the minimum mean square error (MMSE) method were proposed to minimize the mean-square error (MSE) in multiple coordinated VLC attocells under imperfect channel state information (CSI), where all the LED transmitters are assumed to be coordinated through an optical power line communication (PLC) link. Considering multi-user multi-cell multiple-input multiple-output (MIMO) VLC systems, a coordinated zero-forcing (ZF) precoding technique was proposed in \cite{yang2018joint,pham2017multi_1} in a way to cancel the ICI, where the objective was to minimize the MSE in \cite{yang2018joint} and to maximize the achievable sum-rate of the cellular users in \cite{pham2017multi_1}. \\
\indent The integration between CoMP and NOMA was investigated in multi-user multi-cell VLC systems \cite{rajput2019joint,eltokhey2020hybrid}. In \cite{rajput2019joint}, a downlink multi-user multi-cell VLC system was considered and a joint NOMA transmission scheme was proposed, where the users in the overlapping regions are jointly served by all the corresponding VLC APs. In this context, the authors developed two subcarrier allocation techniques, namely, area-based subcarrier allocation and user-based subcarrier allocation. However, the power allocation scheme, which is a crucial factor in the considered system, was not optimized, and suboptimal power allocation coefficients were employed instead. In \cite{eltokhey2020hybrid}, the authors proposed a hybrid NOMA and ZF precoding technique to manage multiple users in multi-cell VLC networks. The proposed approach consists of employing ZF precoding to cancel the ICI at the cell-edge users, while using NOMA to deal with the IUI. However, the authors have ignored the effects of ICI at the cell-center users, which may affect significantly the performance of the system. In fact, in downlink multi-cell VLC systems, the cell-center users may experience ICI as well since typical users may have random orientation and the distance between the APs is typically small, i.e., the corresponding adjacent cells are very close to each other. \\
\indent Based on the above background, one can see that both NOMA and C-NOMA are auspicious multiple access techniques that can boost the network connectivity, while the CoMP technique is an effective ICI mitigation technique. Motivated by this, it is expected that the integration between CoMP and NOMA/C-NOMA techniques can improve the performance of multi-user multi-cell VLC systems and can make the VLC technology a promising candidate for 6G wireless networks. However, against the above background, the optimal power allocation schemes that enhance typical performance metric, such as the sum data rate and the minimum data rate, in multi-user multi-cell VLC systems under the promising CoMP-assisted NOMA and CoMP-assisted C-NOMA techniques were not investigated in the literature, which is the focus of the paper.
\subsection{Contributions and Outcomes}
\indent In this paper, we consider a downlink VLC system consisting of two adjacent VLC attocells that utilize the same frequency resources. Each attocell contains one AP used for illumination and data communication simultaneously. The coverage areas of the two attocells are overlapped. Within this system, three stationary VLC users are communicating simultaneously with the two APs, where each user is equipped with a randomly oriented user-equipment (UE). The first and second UEs are located near the center of each attocell, whereas the third UE is located near the edge of the two attocells, i.e., in the area of overlapping coverage. Therefore, the first and second UEs are associated to the first and second APs, respectively, whereas the third UE is associated to both APs. In this setup, two multiple access techniques are proposed, which are defined as follows.
\begin{itemize}
    \item \textbf{CoMP-assisted NOMA}: Each AP employs NOMA to serve its associated UEs. Hence, the first and second UEs are the strong NOMA UEs in their respective cells, whereas the third UE is the weak UE in both cells. Since the two attocells utilize the same frequency resources, the UEs suffer from ICI. To overcome this issue, the joint-transmission (JT) CoMP technique is employed to mitigate the ICI at the weak UE, whereas the coordinated ZF precoding technique is used to cancel the ICI at the strong UEs.
    \item \textbf{CoMP assisted C-NOMA}: Each AP employs C-NOMA to serve its associated UEs. Hence, the first and second UEs are the strong NOMA UEs in their respective cells, whereas the third UE is the weak UE in both cells. Moreover, the strong UEs have the ability to harvest the energy from the light intensity broadcast from the APs. Therefore, by harnessing the SIC capabilities, the strong UEs can work as relays and forward the decoded weak UE’s signal through RF D2D links using the harvested energy. Furthermore, similar to the CoMP-assisted NOMA scheme, the coordinated ZF precoding technique is used at the two APs to cancel the ICI at the strong UEs.
\end{itemize}

\indent For each proposed multiples access scheme, two optimization frameworks are considered. The first aims to maximize the network sum data rate, under QoS, SIC and power constraints, whereas the second aims to maximize the minimum achievable data rate within the multi-cell VLC network. For each optimization problem, an optimal and low-complexity power control policy is proposed. In the simulation results, the optimality of the proposed low-complexity power control schemes are verified. In addition, three baseline schemes are considered for comparison purposes, namely, the CoMP-assisted OMA scheme, the C-NOMA scheme and the NOMA scheme. The results and benchmarking illustrate the superiority of the proposed schemes in this paper. Finally, the performance of the proposed schemes and the considered baselines are evaluated while varying various system parameters.
\subsection{Outline and Notations}
\indent The rest of this paper is organized as follows. Section II introduces the system model and the proposed multiple access schemes. Section III presents the problems formulations and the proposed solutions. Section IV and V present the simulation results and the conclusion, respectively.\\
\indent Upper case bold characters denote matrices, whereas lower case bold characters denote vectors. $\left(\cdot\right)^T$ and $\left(\cdot\right)^{\perp}$ denote the transpose and the pseudo-inverse operators, respectively. $\left| \left| \cdot \right| \right|_{\infty}$ denotes the infinity norm operator. For every positive real number $a$, the function ${\rm rect}\left( \frac{\cdot}{a} \right)$ denotes the rectangular function within $[0,a]$. For $N \in \mathbb{N}$ and for every real numbers $\left\{a_1, a_2, \dots, a_N\right\}$, $\min \left(a_1, a_2, \dots, a_N \right)$ and $\max \left(a_1, a_2, \dots, a_N \right)$ denote the minimum and the maximum among $\left\{a_1, a_2, \dots, a_N\right\}$, respectively. Finally, for $N \in \mathbb{N}$, $\textbf{0}_N$ and $\textbf{1}_N$ denote the all zeros and all ones $N \times 1$ vectors, respectively.
\section{System Model and Transmission Schemes} 
\subsection{System Model}
\label{subsec:system_model}
\begin{figure}[t]
\centering     
\includegraphics[width=1\linewidth]{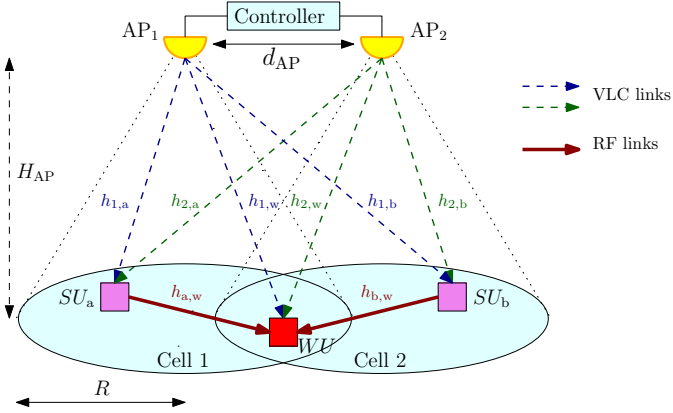}
\caption{System Model.}
\label{fig:IndEnv}
\end{figure}
\indent The system model considered in this paper is shown in Fig.~\ref{fig:IndEnv}, where two APs, each equipped with a set of LEDs, are installed at the ceiling of an indoor environment at a height $H_{\rm AP}$ from the ground.\footnote{In practical VLC systems, the number of deployed APs might be greater than two. In such a case, the APs can be clustered into pairs of adjacent APs and the proposed multiple access schemes can be applied within each pair. In addition, the main reason behind limiting the number of APs to two in this study is that this paper proposes a proof of concept of the CoMP-assisted NOMA and the CoMP-assisted C-NOMA schemes for indoor VLC systems.} Hence, the circular coverage area of each AP has a radius $R = H_{\rm AP} \times \tan \left( \Phi_{1/2} \right)$, where $\Phi_{1/2}$ represents the half-power semi-angle of the LEDs \cite{arfaoui2019snr}. The two APs are jointly monitored by a VLC controller and they share the same frequency bandwidth $B_{\rm v}$. Two UEs, denoted by $SU_{\rm a}$ and $SU_{\rm b}$, respectively, are located around the centers of the first and second cells, respectively, i.e., around the centers of the coverage areas of AP$_1$ and AP$_2$, respectively. On the other hand, since the two APs are adjacent, their resulting coverage areas may be overlapping. In this context, one UE, denoted by $WU$, is located within the edge of both cells, i.e., within the overlapping region between the coverage areas of AP$_1$ and AP$_2$. The area of the resulting overlapping region depends on the radius $R$ of each cell, and hence, depends on the height $H_{\rm AP}$, the half-power semi-angle $\Phi_{1/2}$, and the distance between the two APs, which is denoted by $d_{\rm AP}$. In practical use cases, wide angle LEDs are required in order to provide uniform illumination within the indoor environments. Hence, typical values of the LEDs half-power semi-angle $\Phi_{1/2}$ can go up to $60^\circ$ in practical use cases \cite{LEDhalfpower}. On the other hand, the separation distance between the different APs in practical scenarios depends mainly on the dimensions of the indoor environments as well as the required number of APs that provide the target uniform illuminations. In this context, typical values of the distance between the APs can range from $1$m to $5$m in practical use cases \cite{alfattani2021review}.
\subsection{Transmission Model}
\indent Since they are located around the centers of their associated cells, the UEs $SU_{\rm a}$ and $SU_{\rm b}$ are associated with AP$_1$ and AP$_2$, respectively. However, the UE $WU$ is associated with both APs since it is located within the intersection of their coverage areas. In this considered cellular system, the VLC controller applies NOMA to serve the UEs within each cell, where $SU_{\rm a}$ and $WU$ are the NOMA UEs associated with AP$_1$ and $SU_{\rm b}$ and $WU$ are the NOMA UEs associated with AP$_2$. In this context, when applying the NOMA principle, $SU_{\rm a}$ and $SU_{\rm b}$ are considered as strong UEs in their respective cells since they are located around their centers, whereas $WU$ is considered as the weak UE in both cells. Based on this, the superimposed signals of $SU_{\rm a}$ and $WU$ at AP$_1$ and of $SU_{\rm b}$ and $WU$ at AP$_2$ are expressed, respectively, as
\begin{subequations}
\begin{align}
    s_1 &=  \sqrt{(1-\alpha_1) P_{\rm elec}} s_{\rm a} +  \sqrt{\alpha_1P_{\rm elec}} s_{\rm w}, \\
    s_2 &=  \sqrt{(1-\alpha_2) P_{\rm elec}} s_{\rm b} +  \sqrt{\alpha_2 P_{\rm elec}} s_{\rm w}, 
\end{align}
\end{subequations}
where $s_{\rm a}$, $s_{\rm b}$, and $s_{\rm w}$ represent the messages intended to $SU_{\rm a}$, $SU_{\rm b}$, and $WU$, respectively, such that for all $k \in \left\{\rm a,b,w \right\}$, the message $s_{k} \in \left[ -1, 1\right]$, $P_{\rm elec}$ represents the total electrical power, and $\alpha_1$, $\alpha_2 \in [0,1]$ represent the power allocation factors assigned by AP$_1$ and AP$_2$ to $WU$, respectively, which should be designed by the VLC controller. Afterwards, the VLC controller applies linear precoding to the signals broadcast by AP$_1$ and AP$_2$. As such, the $2\times 1$ vector of optical signals broadcast from the two APs is given by
\begin{equation}
    \label{eq:signal}
    \mathbf{x} = \eta \left(\mathbf{W} \mathbf{s} + I_{\rm DC} \mathbf{1}_2\right),
\end{equation}
where $\mathbf{x} = \left[x_1, x_2 \right]^T$, such that $x_1$ and $x_2$ are the total optical signals broadcast from AP$_1$ and AP$_2$, respectively, $\eta$ [W/A] is the current-to-power conversion efficiency of the LEDs, $\mathbf{W}$ is the $2 \times 2$ precoding matrix of the considered VLC cellular system that should be designed by the VLC controller, $\mathbf{s} = \left[s_1, s_2\right]^T$, and $I_{\rm DC}$ represents the electrical direct-current (DC) provided to each AP. The constant term $I_{\rm DC} \mathbf{1}_2$ is added in \eqref{eq:signal} in order to ensure the positivity of the transmitted signals $\mathbf{W} \mathbf{s}$ at the input of the LEDs. \\
\indent One operating constraint in VLC systems is the peak-power constraint at the LEDs, also known as the amplitude constraint, \cite{arfaoui2018secrecy,arfaoui2018artificial,arfaoui2019secrecy}. In fact, typical LEDs suffer from nonlinear distortion and clipping effects. Hence, in order to maintain a linear current to light conversion and to avoid clipping distortion, a peak-power constraint is imposed on the emitted optical power from the APs \cite{arfaoui2020physical}. This constraint is expressed as
\begin{equation}
    \label{eq:amplitude_constraint}
    ||\mathbf{W} \mathbf{s}||_{\infty}\leq \nu I_{\rm DC},
\end{equation}
where $\nu \in [0,1]$ denotes the modulation index of the VLC system \cite{arfaoui2020physical,mostafa2015physical}. Now, in order to satisfy the constraint in \eqref{eq:amplitude_constraint}, we impose the constraints $||\mathbf{W}||_{\infty} \leq 1$ and $P_{\rm elec} \leq \frac{\left(\nu I_{\rm DC}\right)^2}{2}$ on the precoding matrix $\mathbf{W}$ and the total electrical power $P_{\rm elec}$. In this case, the constraint in \eqref{eq:amplitude_constraint} is satisfied as explained in Appendix \ref{Appendix:A}. 
\subsection{Received Signals}
For all $i \in {1,2}$, let $h_{i,\rm a}$, $h_{i,\rm b}$ denote the positive-valued downlink VLC channel gains from AP$_i$ to $SU_{\rm a}$ and $SU_{\rm b}$, respectively, where the expression of each channel gain can be found in Appendix \ref{Appendix:B}. Therefore, the received signals at $SU_{\rm a}$, $SU_{\rm b}$ can be expressed in a matrix form as \cite{arfaoui2020physical,mostafa2015physical}
\begin{equation}
    \mathbf{y} = R_{\rm p}\eta \mathbf{H}_{\rm a,b}\mathbf{W} \mathbf{s} + R_{\rm p}\eta I_{\rm DC} \mathbf{H} \mathbf{1}_2 + \mathbf{n},
\end{equation}
where $\mathbf{y} = [y_{\rm a},y_{\rm b}]^T$ is the vector of the received signals, such that $y_{\rm a}$ and $y_{\rm b}$ are the received signals at $SU_{\rm a}$ and $SU_{\rm b}$, respectively, $R_{\rm p}$ [V/W] is the responsivity of the PDs of the UEs, and $\mathbf{H}_{\rm a,b}$ denotes the channel matrix between the two APs and the strong UEs, which is expressed as
\begin{equation}
\label{eq:channel_matrix_cell-center}
\mathbf{H}_{\rm a,b} = 
    \begin{bmatrix}
    h_{1,\rm a} & h_{2,\rm a} \\ 
    h_{1,\rm b} & h_{2,\rm b}
    \end{bmatrix}, 
\end{equation}
and $\mathbf{n} = [n_{\rm a},n_{\rm b}]^T$, in which $n_{\rm a}$ and $n_{\rm b}$ represent the additive white Gaussian noise (AWGN) experienced at $SU_{\rm a}$ and $SU_{\rm b}$, respectively. In addition, for all $k \in \left\{\rm a,b \right\}$, the noise $n_k$ is $\mathcal{N}(0,\sigma_{\rm v}^2)$ distributed, where $\sigma_{\rm v}^2 = N_{\rm v} B_{\rm v}$ is the noise power, in which $N_{\rm v}$ is the noise power spectral density. \\
\indent One key parameter in the considered cellular system is the design of the precoding matrix $\mathbf{W}$ at the VLC controller in a way that boosts the overall performance of the system. As discussed above, and as it can be seen from Fig.~\ref{fig:IndEnv}, the considered cellular system suffers from ICI. In fact, since both APs are exploiting the same frequency bandwidth, the strong UE at each cell suffers from ICI that is broadcast from the other cell. Precisely, $SU_{\rm a}$, which is associated to cell 1 and served by AP$_1$, is experiencing ICI generated from AP$_2$ through the wireless channel $h_{2, \rm a}$, and $SU_{\rm b}$, which is associated to cell 2 and served by AP$_2$, is experiencing ICI generated from AP$_1$ through the wireless channel $h_{1, \rm b}$. One way to overcome this issue is to cancel the ICI realizations at both strong UEs through ZF precoding. As such, by taking into account the imposed amplitude constraint $||\mathbf{W}||_{\infty} \leq 1$, the ZF precoding matrix for the considered cellular system can be expressed as $\mathbf{W} = \frac{\mathbf{H}_{\rm a,b}^{\perp}}{\left|\left|\mathbf{H}_{\rm a,b}^{\perp}\right|\right|_{\infty}}$ \cite{pham2017multi_2}. Based on this, the received signals at $SU_{\rm a}$ and $SU_{\rm b}$ are expressed, respectively, as
\begin{subequations} 
\begin{align}
    y_{\rm a} &= \frac{R_{\rm p}\eta \sqrt{(1-\alpha_1) P_{\rm elec}}}{\left|\left|\mathbf{H}_{\rm a,b}^{\perp}\right|\right|_{\infty}} s_{\rm a} + \frac{R_{\rm p}\eta\sqrt{\alpha_1P_{\rm elec}}}{\left|\left|\mathbf{H}_{\rm a,b}^{\perp}\right|\right|_{\infty}} s_{\rm w} \nonumber  \\
    &\,\,+ R_{\rm p}\eta I_{\rm DC}\left(h_{\rm 1,a} + h_{\rm 2,a} \right) + n_{\rm a}, \label{eq:received_signal_cell_center_a} \\
    y_{\rm b} &= \frac{R_{\rm p}\eta \sqrt{(1-\alpha_2) P_{\rm elec}}}{\left|\left|\mathbf{H}_{\rm a,b}^{\perp}\right|\right|_{\infty}} s_{\rm b} + \frac{R_{\rm p}\eta\sqrt{\alpha_2P_{\rm elec}}}{\left|\left|\mathbf{H}_{\rm a,b}^{\perp}\right|\right|_{\infty}} s_{\rm w} \nonumber \\
    &\,\, + R_{\rm p}\eta I_{\rm DC}\left(h_{\rm 1,b} + h_{\rm 2,b} \right) + n_{\rm b} \label{eq:received_signal_cell_center_b}. 
\end{align} 
\end{subequations}
Considering the weak UE, let $h_{i,\rm w}$, for all $i \in \{1,2\}$, denotes the positive-valued downlink VLC channel gains from AP$_i$ to $WU$, where its expression can be found in Appendix \ref{Appendix:B}. In addition, let $\mathbf{h}_{\rm w} = \left[h_{\rm 1, w}, h_{\rm 2,w} \right]^T$ and $\Tilde{\mathbf{h}}_{\rm w} = \left[\Tilde{h}_{\rm 1, w}, \Tilde{h}_{\rm 2,w} \right]^T = \mathbf{W}^T \mathbf{h}_{\rm w}$. Based on this, the received signal at $WU$ through the direct transmission from the two APs is expressed as
\begin{equation} 
\label{eq:received_signal_cell_edge}
\begin{split}
    y_{\rm w} &= R_{\rm p}\eta \mathbf{h}_{\rm w}^T\mathbf{W} \mathbf{s} + R_{\rm p}\eta I_{\rm DC} \mathbf{h}_{\rm w}^T \mathbf{1}_2 + n_{\rm w}, \\
    &=\frac{R_{\rm p}\eta \sqrt{(1-\alpha_1) P_{\rm elec}}}{\left|\left|\mathbf{H}_{\rm a,b}^{\perp}\right|\right|_{\infty}} \Tilde{h}_{\rm 1,w} s_{\rm a} + \frac{R_{\rm p}\eta\sqrt{(1-\alpha_2)P_{\rm elec}}}{\left|\left|\mathbf{H}_{\rm a,b}^{\perp}\right|\right|_{\infty}} \Tilde{h}_{\rm 2,w} s_{\rm b} \\ 
    &\quad+ \frac{R_{\rm p}\eta\sqrt{P_{\rm elec}}}{\left|\left|\mathbf{H}_{\rm a,b}^{\perp}\right|\right|_{\infty}} \left(\Tilde{h}_{\rm 1,w}\sqrt{\alpha_1} + \Tilde{h}_{\rm  2,w}\sqrt{\alpha_2}\right)s_{\rm w} \\
    &\quad + R_{\rm p}\eta I_{\rm DC}\left(h_{\rm 1,w} + h_{\rm 2,w} \right) +  n_{\rm w},
\end{split} 
\end{equation}
where $n_{\rm w}$ represents the AWGN experienced at $WU$ that is $\mathcal{N}(0,\sigma_{\rm v}^2)$ distributed.
\subsection{Data Rate Analysis}
\subsubsection{CoMP-Assisted NOMA} 
Since VLC systems impose an amplitude constraint (or a peak-power constraint) on the input signals, the Gaussian distribution is not admissible for VLC signals. Due to this, the capacity of VLC systems remains unknown \cite{lapidoth2009capacity}. Extensive studies have put much effort into obtaining closed-form expressions for the achievable data rate, and have demonstrated that it can be expressed as $\frac{1}{2} \log \left(1 + c \times \text{SNR} \right)$, where $c = \frac{1}{2 \pi e}$ and $e$ is the Euler constant, when the visible light signals follow the truncated Gaussian distribution \cite{chaaban2016capacity,zhou2019bounds}. On the other hand, the data rates received at the UEs are governed by the NOMA technique employed by the APs. Accordingly, $SU_{\rm a}$ and $SU_{\rm b}$ will first use SIC to decode the message of $WU$ and then decode their own messages. Based on this, and according to \cite{chaaban2016capacity}, the achievable data rate of $SU_{\rm a}$ to decode the signal of $WU$ is expressed as
\begin{equation}
    \label{eq:a_w_rate}
     R_{\rm a \rightarrow w}(\alpha_1) = \frac{B_{\rm v}}{2} \log \left( 1 + \frac{c \alpha_1}{c \left(1-\alpha_1\right) + \frac{1}{\gamma_{RX}}} \right),
\end{equation}
where $\gamma_{Rx} = \frac{R_{\rm p}^2\eta^2P_{\rm elec}\sigma_{\rm s}^2}{\left|\left|\mathbf{H}_{\rm a,b}^{\perp}\right|\right|_{\infty}^2 \sigma_{\rm v}^2}$, in which $\sigma_{\rm s}^2 = \sigma_{\rm d}^2 - \frac{\sigma_{\rm d} \exp \left(\frac{-1}{2 \sigma_{\rm d}^2} \right)}{\erf \left( \frac{1}{\sigma_{\rm d} \sqrt{2}} \right)}$, such that $\sigma_{\rm d}^2$ is the scale parameter of the truncated Gaussian distribution of the input signals $s_{\rm a}$, $s_{\rm b}$ and $s_{\rm w}$, respectively. After that $SU_{\rm a}$ performs SIC and cancels the message of $WU$ from its reception, its achievable data rate to decode its own message is expressed as
\begin{equation}
\label{eq:a_rate}
    R_{\rm a}(\alpha_1) = \frac{B_{\rm v}}{2} \log \left( 1 + c \gamma_{RX} \left(1-\alpha_1\right)\right).
\end{equation}
Similarly, the achievable data rate of $SU_{\rm b}$ to decode the signal of $WU$ is expressed as
\begin{equation}
\label{eq:b_w_rate}
     R_{\rm b \rightarrow w}(\alpha_2) = \frac{B_{\rm v}}{2} \log \left( 1 + \frac{c \alpha_2}{c \left(1-\alpha_2\right) + \frac{1}{\gamma_{RX}}} \right),
\end{equation}
and its achievable data rate to decode its own message after performing SIC is expressed as
\begin{equation}
    \label{eq:b_rate}
    R_{\rm b}(\alpha_2) = \frac{B_{\rm v}}{2} \log \left(1 + c \gamma_{RX} \left( 1-\alpha_2\right) \right).
\end{equation}
\indent For $WU$, and following the NOMA principle, it will treat the messages of $SU_{\rm a}$ and $SU_{\rm b}$ as noise and will decode directly its own message. Hence, the achievable data rate of $WU$ to decode its own message is expressed as shown in \eqref{rate_w} on top of next page.
\begin{figure*}[t]
\begin{equation}
    \label{rate_w}
    \begin{split}
        &R_{\rm w \rightarrow w}^{\rm VL}(\alpha_1,\alpha_2) = \frac{B_{\rm v}}{2} \log \left(1 + \frac{c(\Tilde{h}_{\rm 1,w}\sqrt{\alpha_1} + \Tilde{h}_{\rm 2,w}\sqrt{\alpha_2})^2}{c\Tilde{h}_{\rm 1,w}^2(1-\alpha_1) + c\Tilde{h}_{\rm 2,w}^2(1-\alpha_2) + \frac{1}{\gamma_{\rm RX}}} \right).
    \end{split}
\end{equation}
\noindent\makebox[\linewidth]{\rule{\textwidth}{0.4pt}}
\end{figure*}
Finally, to ensure a successful SIC at the strong UEs, the $WU$’s message should be detectable at each strong UE \cite{jiao2020max,jiao2021max,timotheou2015fairness,hanif2015minorization,zhang2016robust}. Thus, the achievable data rate of $WU$ is given by \cite{jiao2020max,jiao2021max,timotheou2015fairness,hanif2015minorization,zhang2016robust}
\begin{equation}
    \label{rate_w_combined}
    R_{\rm w}^{\rm VL} = \min \left(R_{\rm a \rightarrow w}(\alpha_1), R_{\rm w \rightarrow w}^{\rm VL}(\alpha_1,\alpha_2),R_{\rm b \rightarrow w}(\alpha_2)\right).
\end{equation}
\subsubsection{CoMP-Assisted C-NOMA} In this scheme, and similar to the CoMP-assisted NOMA scheme, each strong UE performs SIC and cancels the message of the weak UE from its reception. In this context, the achievable data rates of $SU_{\rm a}$ and $SU_{\rm b}$ to decode the signal of $WU$ and to decode their own signals are given in \eqref{eq:a_w_rate}-\eqref{eq:b_rate}. However, different from the CoMP-assisted NOMA scheme, it is assumed that the strong NOMA UEs $SU_{\rm a}$ and $SU_{\rm b}$ can work as relays that have the ability to harvest the energy from the light intensity broadcast from the APs and then to utilize it to forward the decoded weak UE’s signal. To harvest the energy, a capacitor separates the DC component from the received electrical signal at each strong UE and forwards it to its energy harvesting circuit \cite{wang2015design,diamantoulakis2018simultaneous}. The received energy at the the strong user $k \in \left\{ \rm a, b \right\}$ is given by \cite{li2011solar} 
\begin{equation}
    E_{k} = f V_t I_{{\rm DC},k}^{\rm r} \log \left[1 + \frac{I_{{\rm DC},k}^{\rm r}}{I_0} \right],
\end{equation}
where $I_{{\rm DC},k}^{\rm r} = R_{\rm p}\eta I_{\rm DC}\left(h_{1,k} + h_{2,k} \right)$ is the received DC at the strong user $k$, $V_t$ is the thermal voltage, $f$ is the fill factor, and $I_0$ is the dark saturation current of the PD. At this stage, each transmission time-slot is divided into two equidistant mini-time slots. In the first time slot, each strong UE harvest the energy and charge its battery, whereas in the second mini-time slot, it discharge the harvested energy and transmit the data of $WU$ through an RF link. In this case, for all $k \in \left\{ \rm a, b \right\}$, the RF transmission power from $SU_{k}$ is given by $P_k^{\rm RF} = E_{k}$ and the achievable data rate of $WU$ that can be offered by the strong UEs $SU_{\rm a}$ and $SU_{\rm b}$ through the RF D2D links is given by 
\begin{equation}
   R_{\rm a,b \rightarrow w}^{\rm RF} = \frac{B_{\rm r}}{2} \log \left(1 + \frac{|g_{\rm a,w}|^2 P_{\rm a}^{\rm RF} + |g_{\rm b,w}|^2 P_{\rm b}^{\rm RF}}{\sigma_{\rm r}^2} \right),
\end{equation}
where, for all $k \in \left\{ \rm a, b \right\}$, $g_{k,\rm w}$ is the RF channel coefficient between the strong UE $SU_k$ and $WU$ and $\sigma_{\rm r}^2 = N_{\rm r} B_{\rm r}$ is the RF noise power, in which $N_{\rm r}$ is the power spectral density of the RF noise and $B_{\rm r}$ is the RF modulation bandwidth. Using the same RF D2D channel model in indoor environments adopted in \cite{xiao2019hybrid}, the RF channel coefficients between the strong UEs and the weak UE are modeled by
the Nakagami-$m$ fading channel, i.e., for all $k \in \{\rm a,b\}$, the RF channel coefficient $g_{k,\rm w} \sim Nakagami\left(F, d_{k,w}^{-\mu}\right)$, where $F$ represents the fading parameter, $d_{k,w}$ represents the Euclidean distance between the strong UE $SU_{k}$ and the weak UE $WU$, and $\mu$ represents the path-loss exponent. \\ 
\indent The weak UE $WU$ is receiving two copies of his message through two different links, one from $SU_{\rm a}$ and one from $SU_{\rm b}$. Hence, $WU$ can combine these two copies using the maximum-ratio-combining (MRC) technique and then decodes his own message, which results in the achievable data rate $R_{\rm a,b \rightarrow w}^{\rm RF}$. However, this data rate is achievable if and only if the copies transmitted from the strong UEs are truly the exact message $s_{\rm w}$ of $WU$, i.e., if and only if the strong UEs are able to decode the message $s_{\rm w}$ of $WU$. Therefore, the data rate achieved at the weak UE is constrained by the achievable data rates of the strong UEs to decode the message $s_{\rm w}$ of the $WU$, which are the data rates $R_{\rm a \rightarrow w}(\alpha_1)$ and $R_{\rm b \rightarrow w}(\alpha_2)$. Consequently, in line with the results of \cite{elhattab2020joint,elhattab2022joint}, the resulting achievable data rate of the weak UE $WU$ from the cooperative diversity of the strong UEs is expressed as
\begin{equation}
    \label{eq:rate_weak_CNOMA}
    R_{\rm w}^{\rm RF}(\alpha_1,\alpha_2) = \min \left(R_{\rm a \rightarrow w}(\alpha_1), R_{\rm a,b \rightarrow w}^{\rm RF}, R_{\rm b \rightarrow w}(\alpha_2)\right).
\end{equation}
\section{Problems Formulation and Proposed Solutions}
In this section, we investigate the sum data rate and the minimum data rate maximization problems for the considered VLC cellular system for both the CoMP-assisted NOMA and the CoMP-assisted C-NOMA schemes. The sum data rate maximization problem is investigated in subsection \ref{Sec:Sum_Rate}, whereas the minimum data rate maximization problem is considered in subsection \ref{Sec:Min_Rate}.
\subsection{Sum data rate Maximization}
\label{Sec:Sum_Rate}
The first objective is to maximize the sum data rate of the considered VLC cellular system under both the CoMP-assisted NOMA and the CoMP-assisted C-NOMA schemes, while a target QoS should be guaranteed for each UE in terms of its required data rate threshold, denoted by $R_{\rm th}$.
\subsubsection{CoMP-Assisted NOMA} For this scheme, the preset objective can be reached by solving the following optimization problem.
\allowdisplaybreaks
\begingroup
\begin{subequations}
\label{OPT_1}
\begin{align}
\mathcal{P}_1:\,\, &R_{\rm sum}^{{\rm VL}^*}=\max_{\alpha_1, \alpha_2}\,\, R_{\rm a}(\alpha_1)+R_{\rm b}(\alpha_2)+R_{\rm w}^{\rm VL}(\alpha_1,\alpha_2),\label{Prob1} \\
&\text{s.t.}\qquad \,\,\,\,  0 \leq \alpha_i \leq 1, \quad \forall\,\, i \in \{1,2\}, \label{Const:C1}\\
&\qquad \qquad  R_{\rm a \rightarrow w}(\alpha_1) \geq R_{\rm th}, \label{Const:C3}\\
&\qquad \qquad  R_{\rm a}(\alpha_1) \geq R_{\rm th}, \label{Const:C4}\\
&\qquad \qquad  R_{\rm b \rightarrow w}(\alpha_2) \geq R_{\rm th}, \label{Const:C5}\\
&\qquad \qquad  R_{\rm b}(\alpha_2) \geq R_{\rm th}, \label{Const:C6}\\
&\qquad \qquad  R_{\rm w \rightarrow w}^{\rm VL} \geq R_{\rm th}. \label{Const:C7}
\end{align}
\end{subequations}
\endgroup
Based on the data rate expressions presented in \eqref{eq:a_w_rate}-\eqref{rate_w}, problem $\mathcal{P}_1$ is a non-linear non-convex problem that cannot be solved in a straightforward manner. Alternatively, we propose in the following an efficient and low complexity approach to solve problem $\mathcal{P}_1$. First, in order to be able to solve problem $\mathcal{P}_1$, the conditions under which at least one feasible solution exists must be derived. In this context, the feasibility conditions of problem $\mathcal{P}_1$ are presented in the following theorem.
\begin{theorem}
Problem $\mathcal{P}_1$ is feasible if and only if the following conditions hold: 
\begin{subequations}
\label{eq:cdts_comp_noma}
\begin{align}
    &\text{Condition 1:} \,\, \alpha_{\min} \leq \alpha_{\max}, \label{eq:cdt11}\\
    &\text{Condition 2:} \,\,  0 \leq g\left(\alpha_{\max},\alpha_{\max}\right), \label{eq:cdt12}
\end{align}
\end{subequations}
where $\alpha_{\min} = = \frac{t_{\rm v}\left(c \gamma_{\rm RX} + 1 \right)}{c \gamma_{\rm RX} \left(1+t_{\rm v} \right)}$ and $\alpha_{\max} = 1 - \frac{t_{\rm v}}{c \gamma_{\rm RX}}$, such that $t_{\rm v} = \exp \left(\frac{2 R_{\rm th}}{B_{\rm v}} \right) - 1$, and the function $g\left(\cdot,\cdot\right)$ is expressed, for all $\left(\alpha_1,\alpha_2\right) \in \mathbb{R}^2$, as
\begin{equation}
    \begin{split}
         g\left(\alpha_1,\alpha_2\right) &= c\left(1+t_{\rm v} \right)\Tilde{h}_{\rm 1,w}^2\alpha_1 + c\left(1+t_{\rm v} \right)\Tilde{h}_{\rm 2,w}^2\alpha_2 \\
         &+ c\Tilde{h}_{\rm 1,w}\Tilde{h}_{\rm 2,w} \sqrt{\alpha_1 \alpha_2} - t \left(c\Tilde{h}_{\rm 1,w}^2 + c\Tilde{h}_{\rm 2,w}^2 + \frac{1}{\gamma_{\rm RX}} \right).
    \end{split}
\end{equation}
\end{theorem}
\begin{proof}
See Appendix \ref{Appendix: Theorem 1}.
\end{proof}
Based on \textbf{Theorem 1} and its proof in Appendix \ref{Appendix: Theorem 1}, the feasibility region of the optimization problem $\mathcal{P}_1$ is defined by the set $\left\{\left(\alpha_1,\alpha_2\right) \in \left[\alpha_{\min},\alpha_{\max} \right]^2 \big| g\left(\alpha_1,\alpha_2\right) \geq 0 \right\}$. Now that the feasibility conditions are set, our objective is to find the optimal solution of problem $\mathcal{P}_1$, i.e., the optimal values of the power allocation fractions $\alpha_1$ and $\alpha_2$ that maximize the network sum data rate $R_{\rm sum}^{\rm VL} = R_{\rm a}+R_{\rm b}+ R_{\rm w}^{\rm VL}$. In this setup, since $SU_{\rm a}$ and $SU_{\rm b}$ are the strong UEs and $WU$ is the weak UE, the optimal power allocation strategy that maximizes the network sum data rate is the one that allocates the lowest possible power to the weak UE $WU$ while guaranteeing its required data rate threshold $R_{\rm th}$, and the remaining of the power to the strong UEs $SU_{\rm a}$ and $SU_{\rm b}$ \cite{huu2020low,elhattab2022joint}. Therefore, since the expressions of the achievable data rates $R_{\rm a}$ and $R_{\rm b}$ are decreasing functions of $\alpha_1$ and $\alpha_2$, respectively, and the expression of the achievable data rate $R_{\rm w}^{\rm VL}$ is an increasing function with respect to $\alpha_1$ and $\alpha_2$, then the optimal power allocation strategy is the one that satisfies the inequality $g\left(\alpha_1,\alpha_2\right) \geq 0$ with the lowest possible values of $\alpha_1$ and $\alpha_2$ within the square $\left[\alpha_{\min},\alpha_{\max} \right]^2$. \\
\indent In order to determine the optimal power allocation coefficients $\left(\alpha_{1}^*,\alpha_{2}^*\right)$, we opt for a discrete line search technique within the segment $[\alpha_{\min}, \alpha_{\max}]$. Let $K \in \mathbb{N}$ be the number of discrete points within $[\alpha_{\min}, \alpha_{\max}]$. Based on this, the discrete line search technique within the segment $[\alpha_{\min}, \alpha_{\max}]$ works as follows. For all $i \in \llbracket 0, K-1 \rrbracket$, we calculate $\alpha_1^i = \alpha_{\min} + \frac{\alpha_{\max}-\alpha_{\min}}{K-1}\times i$. Then, we determine the lowest value of $\alpha_2^i$ that satisfies the inequality $g(\alpha_{1}^i,\alpha_{2}^i) \geq 0$ using the approach presented in Appendix \ref{Appendix: line search sum rate}. Afterwards, we calculate the corresponding network sum data rate $R_i(\alpha_1^i,\alpha_2^i) = R_{\rm a}(\alpha_1^i)+R_{\rm b}(\alpha_2^i)+R_{\rm w}^{\rm VL}(\alpha_1^i,\alpha_2^i)$. Finally, the optimal power allocation fractions $(\alpha_1^*,\alpha_2^*)$, solution of problem $\mathcal{P}_1$, is the one that achieves the highest network sum data rate, i.e., 
\begin{equation}
    (\alpha_1^*,\alpha_2^*) = \underset{(\alpha_1^i,\alpha_2^i)}{\rm argmax}\,\, R_i(\alpha_1^i,\alpha_2^i),
\end{equation}
which can be obtained through a brute force search over the set $\left\{R_i(\alpha_1^i,\alpha_2^i) \Big| i \in \llbracket 0, K-1 \rrbracket \right\}$. As it can be seen, the proposed solution approach is based on a discrete line search technique over a set of $K$ points, which has a complexity of $\mathcal{O}(K)$, i.e., a linear complexity. This fact demonstrates the low complexity of the proposed solution approach.
\subsubsection{CoMP-Assisted C-NOMA} Under this scheme, the objective of maximizing the sum data rate of the considered VLC cellular system while guaranteeing the required data rate threshold $R_{\rm th}$ at each UE can be reached by solving the following optimization problem.
\allowdisplaybreaks
\begingroup
\begin{subequations}
\label{OPT_2}
\begin{align}
\mathcal{P}_2:\,\, &R_{\rm sum}^{{{\rm VL}/{\rm RF}}^*}=\max_{\alpha_1, \alpha_2}\,\, R_{\rm a}(\alpha_1)+R_{\rm b}(\alpha_2)+R_{\rm w}^{\rm RF}(\alpha_1,\alpha_2),\label{Prob2} \\
&\text{s.t.}\qquad \,\,\,\,  \eqref{Const:C1}-\eqref{Const:C6},\\
&\qquad \qquad  R_{\rm a,b \rightarrow w}^{\rm RF} \geq R_{\rm th}.
\end{align}
\end{subequations}
\endgroup
Based on the rate expressions presented in \eqref{eq:a_w_rate}-\eqref{eq:b_rate}, problem $\mathcal{P}_2$ is a non-linear non-convex problem that cannot be solved in a straightforward manner. Alternatively, we propose in the following an efficient and low complexity approach to solve problem $\mathcal{P}_2$. First, and similar to the CoMP-assisted NOMA scheme, the conditions under which at least one feasible solution for problem $\mathcal{P}_2$ exists must be derived. In this context, the feasibility conditions of problem $\mathcal{P}_2$ are presented in the following theorem.
\begin{theorem}
Problem $\mathcal{P}_2$ is feasible if and only if the following conditions hold: 
\begin{subequations}
\label{eq:cdts_comp_cnoma}
\begin{align}
    &\text{Condition 1:} \,\, \alpha_{\min} \leq \alpha_{\max}, \label{eq:cdt21}\\
    &\text{Condition 2:} \,\,  t_{\rm r} \leq \frac{|g_{\rm a,w}|^2 P_{\rm a}^{\rm RF} + |g_{\rm b,w}|^2 P_{\rm b}^{\rm RF}}{\sigma_{\rm r}^2}, \label{eq:cdt22}
\end{align}
\end{subequations}
where $t_{\rm r} = \exp \left(\frac{2 R_{\rm th}}{B_{\rm r}} \right) - 1$.
\end{theorem}
\begin{proof}
The proof can be easily deducted from the one of \textbf{Theorem 1} in Appendix \ref{Appendix: Theorem 1}.
\end{proof}
It is important to mention that the condition in \eqref{eq:cdt22} is related to the RF relaying links from the strong UEs to the weak UE and is not related to the power allocation fractions $\left(\alpha_1, \alpha_2\right)$. Specifically, the RF relaying links are governed by the RF channel conditions, represented by the channel coefficients $g_{\rm a,w}$ and $g_{\rm a,b}$, the electrical powers $P_{\rm a}^{\rm RF}$ and $P_{\rm b}^{\rm RF}$ that are resulting from the harvested optical powers at the strong UEs, the available RF bandwidth $B_{\rm r}$ and the electrical noise power $\sigma_r^2$ at the $WU$. When combined together, these parameters have to make the resulting RF data rate at $WU$ that is given by $R_{\rm a,b \rightarrow w}^{\rm RF}$ greater than the required data rate threshold $R_{\rm th}$, which is the focus of the condition in \eqref{eq:cdt22}. In this context, assuming that the condition in \eqref{eq:cdt22} is satisfied, and based on \textbf{Theorem 2}, the feasibility region of problem $\mathcal{P}_2$ is the square $\left[\alpha_{\min},\alpha_{\max} \right]^2$.\\
\indent Now that the feasibility conditions are set, our objective is to find the optimal solution of problem $\mathcal{P}_2$, i.e., the optimal values of the power allocation fractions $\alpha_1$ and $\alpha_2$ that maximize the network sum data rate $R_{\rm sum}^{{{\rm VL}/ {\rm RF}}} = R_{\rm a}+R_{\rm b}+R_{\rm w}^{\rm RF}$. In this context, $R_{\rm sum}^{{{\rm VL}/{\rm RF}}^*}$ can be expressed as in \eqref{eq:rate_sum} on top of next page.
\begin{figure*}
\begin{subequations}
\label{eq:rate_sum}
\begin{align}
R_{\rm sum}^{{{\rm VL}/{\rm RF}}^*}&=\max_{\alpha_1, \alpha_2}\,\, \left[R_{\rm a}(\alpha_1)+R_{\rm b}(\alpha_2)+R_{\rm w}^{\rm RF}(\alpha_1,\alpha_2)\right],\\
&=\max_{\alpha_1, \alpha_2} \,\, \left[R_{\rm a}(\alpha_1)+R_{\rm b}(\alpha_2)+ \min \left(R_{\rm a \rightarrow w}(\alpha_1), R_{\rm a,b \rightarrow w}^{\rm RF}, R_{\rm b \rightarrow w}(\alpha_2)\right)\right],\\
&= \max_{\alpha_1, \alpha_2} \,\, \min \bigg( R_{\rm a}(\alpha_1)+R_{\rm b}(\alpha_2)+R_{\rm a \rightarrow w}(\alpha_1), R_{\rm a}(\alpha_1)+R_{\rm b}(\alpha_2)+R_{\rm a,b \rightarrow w}^{\rm RF}, R_{\rm a}(\alpha_1)+R_{\rm b}(\alpha_2)+R_{\rm b \rightarrow w}(\alpha_2)\bigg) \nonumber\\
&= \max_{\alpha_1, \alpha_2} \,\, \min \bigg(R_{\rm b}(\alpha_2) +  \frac{B_{\rm v}}{2} \log \left(1 + c\gamma_{\rm RX} \right), R_{\rm a}(\alpha_1)+R_{\rm b}(\alpha_2)+R_{\rm a,b \rightarrow w}^{\rm RF}, R_{\rm a}(\alpha_1)+\frac{B_{\rm v}}{2} \log \left(1 + c\gamma_{\rm RX} \right)\bigg).
\end{align}
\end{subequations}
\noindent\makebox[\linewidth]{\rule{\textwidth}{0.4pt}}
\end{figure*}
Since the expressions of the achievable data rates $R_{\rm a}(\alpha_1)$ and $R_{\rm b}(\alpha_2)$ are decreasing functions with respect to $\alpha_1$ and $\alpha_2$, respectively, the first term inside the $\min$ operator in \eqref{eq:rate_sum} is maximized when $\alpha_2 = \alpha_{\min}$, the second is maximized when $\left(\alpha_1,\alpha_2\right) = \left(\alpha_{\min},\alpha_{\min}\right)$, and the third term is maximized when $\alpha_1 = \alpha_{\min}$. Based on this, we conclude that the optimal solution of problem $\mathcal{P}_2$ is $\left(\alpha_1^*,\alpha_2^*\right) = \left(\alpha_{\min},\alpha_{\min}\right)$.
\subsection{Minimum data rate Maximization}
\label{Sec:Min_Rate}
The second objective of this paper is to maximize the minimum data rate of the UEs within the considered VLC cellular system under the CoMP-assisted NOMA and the CoMP-assisted C-NOMA schemes.
\subsubsection{CoMP-Assisted NOMA} For this scheme, and by recalling the expression of $R_{\rm w}^{\rm VL}$ in \eqref{rate_w_combined}, the target objective can be reached by solving the optimization problem $\mathcal{P}_3$ in \eqref{OPT_3} on top of next page.
\begin{figure*}
    \begin{subequations}
\label{OPT_3}
\begin{align}
\mathcal{P}_3:\,\, &R_{\rm min}^{{\rm VL}^*}=\max_{\alpha_1, \alpha_2}\,\, \left[ \min \left\{R_{\rm a \rightarrow w}(\alpha_1), R_{\rm a}(\alpha_1), R_{\rm b \rightarrow w}(\alpha_2) , R_{\rm b}(\alpha_2) , R_{\rm w \rightarrow w}^{\rm VL}(\alpha_1,\alpha_2) \right\} \right],\label{Prob3} \\
&\text{s.t.}\qquad \,\,\,\,  0 \leq \alpha_i \leq 1, \quad \forall\,\, i \in \{1,2\}.
\end{align}
\end{subequations}
\noindent\makebox[\linewidth]{\rule{\textwidth}{0.4pt}}
\end{figure*}
Based on the rate expressions presented in \eqref{eq:a_w_rate}-\eqref{rate_w}, problem $\mathcal{P}_3$ is a non-linear non-convex problem that cannot be solved in a straightforward manner. Alternatively, we propose in the following an efficient and low complexity approach to solve problem $\mathcal{P}_3$, which is feasible over the entire set $\left[0,1 \right]^2$. First, it can be easily demonstrated that, for $\alpha_1, \alpha_2 \in \left[0, 1 \right]$, we have 
\begin{subequations}
\label{eq:rate_equality_a_b_w}
\begin{align}
    &R_{\rm a \rightarrow w}(\alpha_1) \geq R_{\rm a}(\alpha_1) \Leftrightarrow  \alpha_1 \geq \alpha_0,\\
    &R_{\rm b \rightarrow w}(\alpha_2) \geq R_{\rm b}(\alpha_2) \Leftrightarrow  \alpha_2 \geq \alpha_0,
\end{align}
\end{subequations}
where $\alpha_0 = \frac{\left(c \gamma_{\rm RX}+1\right)-\sqrt{c \gamma_{\rm RX}+1}}{c \gamma_{\rm RX}}$. Consequently, for all $\left(\alpha_1, \alpha_2 \right) \in [\alpha_{0}, 1]^2$, we have
\begin{equation}
    \begin{split}
        &\min \left\{R_{\rm a \rightarrow w}(\alpha_1), R_{\rm a}(\alpha_2), R_{\rm b \rightarrow w}(\alpha_1) , R_{\rm b}(\alpha_2) , R_{\rm w \rightarrow w}^{\rm VL} \right\} \\
        &= \min \left\{ R_{\rm a}(\alpha_2), R_{\rm b}(\alpha_2), R_{\rm w \rightarrow w}^{\rm VL} \right\}.
    \end{split}
\end{equation}
At this stage, we apply the discrete line search approach within the segment $\left[\alpha_0, 1 \right]$ in order to maximize $\min \left\{ R_{\rm a}(\alpha_2), R_{\rm b}(\alpha_2), R_{\rm w \rightarrow w}^{\rm VL} \right\}$. Specifically, let $K \in \mathbb{N}$ be the number of discrete points within $[\alpha_{0}, 1]$. Based on this, the discrete line search technique within the segment $[\alpha_{0}, 1]$ works as follows. For all $i \in \llbracket 0, K-1 \rrbracket$, we calculate $\alpha_1^i = \alpha_{0} + \frac{1-\alpha_{0}}{K-1}\times i$. Then, we determine the lowest value of $\alpha_2^i \in \left[ \alpha_0,1\right]$ that satisfies the inequality $R_{\rm w \rightarrow w}^{\rm VL}(\alpha_1^i,\alpha_2^i) \geq R_{\rm b}(\alpha_2^i)$ using the same approach of the sum data rate maximization problem presented in Appendix \ref{Appendix: line search sum rate}. Afterwards, we calculate the corresponding minimum data rate $R_{\rm min}^i(\alpha_1^i, \alpha_2^i) = \min \left\{R_{\rm a \rightarrow w}(\alpha_1^i), R_{\rm a}(\alpha_1^i), R_{\rm b \rightarrow w}(\alpha_2^i) , R_{\rm b}(\alpha_2^i) , R_{\rm w \rightarrow w}^{\rm VL}(\alpha_1^i,\alpha_2^i) \right\}$. Finally, the optimal power allocation fractions $(\alpha_1^*,\alpha_2^*)$, solution of problem $\mathcal{P}_3$, is the one that achieves the highest minimum data rate, i.e., 
\begin{equation}
    (\alpha_1^*,\alpha_2^*) = \underset{(\alpha_1^i,\alpha_2^i)}{\rm argmax}\,\, R_{\min}^i(\alpha_1^i,\alpha_2^i),
\end{equation}
which can be obtained through a brute force search over the set $\left\{R_{\min}^i(\alpha_1^i,\alpha_2^i) \Big| i \in \llbracket 0, K-1 \rrbracket \right\}$. As it can be seen, the proposed solution approach is based on a discrete line search technique over a set of $K$ points,  which has a complexity of $\mathcal{O}(K)$, i.e., a linear complexity, This fact demonstrates the low complexity of the proposed solution approach.
\subsubsection{CoMP-Assisted C-NOMA} Since the achievable data rate $R_{\rm a,b \rightarrow w}^{\rm RF}$ is related to the RF relaying links from the strong UEs to the weak UE and is not related to the power allocation fractions $\left(\alpha_1, \alpha_2\right)$, the maximum minimum data rate of the UEs for the considered VLC cellular system under the CoMP-assisted C-NOMA scheme is given by $R_{\rm min}^{{\rm VL}/{\rm RF}^*} = \min \left\{R_{\rm a,b \rightarrow w}^{\rm RF},R_{\rm min}^{{\rm a,b}^*} \right\}$, where $R_{\rm min}^{{\rm a,b}^*}$ is obtained by solving the optimization problem $\mathcal{P}_4$ on top of next page.
\begin{figure*}
\begin{subequations}
\label{OPT_4}
\begin{align}
\mathcal{P}_4:\,\, &R_{\rm min}^{{\rm a,b}^*}=\max_{\alpha_1, \alpha_2}\,\, \left[ \min \left\{R_{\rm a \rightarrow w}(\alpha_1), R_{\rm a}(\alpha_1), R_{\rm b \rightarrow w}(\alpha_2) , R_{\rm b}(\alpha_2) \right\} \right],\label{Prob4} \\
&\text{s.t.}\qquad \,\,\,\,  0 \leq \alpha_i \leq 1, \quad \forall\,\, i \in \{1,2\}.
\end{align}
\end{subequations}
\noindent\makebox[\linewidth]{\rule{\textwidth}{0.4pt}}
\end{figure*}
It can be easily verified that problem $\mathcal{P}_{4}$ is a non-linear non-convex problem, due to the non-convexity of the data rate expressions in \eqref{eq:a_w_rate}-\eqref{eq:b_rate}, hence, it can not be solved in a straightforward manner. Alternatively, we propose an efficient and low complexity approach to solve problem $\mathcal{P}_4$, which is feasible over the entire set $\left[0,1 \right]^2$. It can be easily noticed that the maximum value of $\min \left\{R_{\rm a \rightarrow w}(\alpha_1), R_{\rm a}(\alpha_1), R_{\rm b \rightarrow w}(\alpha_2) , R_{\rm b}(\alpha_2) \right\}$ is achieved if and only if $R_{\rm a \rightarrow w}(\alpha_1) = R_{\rm a}(\alpha_1)$ and $R_{\rm b \rightarrow w}(\alpha_2) = R_{\rm b}(\alpha_2)$, which, based on \eqref{eq:rate_equality_a_b_w}, is reached if and only if $\left(\alpha_1, \alpha_2 \right) = \left(\alpha_0,\alpha_0 \right)$. Consequently, the optimal solution of problem $\mathcal{P}_{4}$ is $\left(\alpha_1^*, \alpha_2^* \right) = \left(\alpha_0,\alpha_0 \right)$.
\section{Simulation Results}
In this section, our objective is to evaluate the performance of the proposed CoMP-assisted NOMA and CoMP-assisted C-NOMA schemes for the considered indoor VLC cellular system through extensive simulations.
\subsection{Simulations Settings}
\indent We consider an indoor environment with length $7$m and width $7$m \cite{eltokhey2020hybrid}, in which the VLC cellular system presented in Section \ref{subsec:system_model} is deployed. Each AP is oriented vertically downward, whereas each UE has a random orientation that is generated using the measurements-based orientation models proposed in \cite{soltani2018modeling,soltani2019bidirectional,arfaoui2020measurements}. Moreover, the first and second UEs are randomly located around the center of the first and the second attocells, respectively, whereas the third UE is randomly located near the edge of both attocells. Therefore, using the proposed CoMP-assisted NOMA and CoMP-assisted C-NOMA schemes, the first and second UEs are associated to the first and second attocells, respectively, and each one of them is communicating with the AP of its associated attocell, whereas the third UE is associated to the two attocells. Unless otherwise stated, the simulation parameters and settings used throughout the paper are shown in Table \ref{T2}. All the results are generated through $10^4$ independent Monte Carlo trials, where in each trial the locations and the orientations of the UEs are generated randomly as discussed above. Finally, for comparison purposes, three baselines are considered, which are
\begin{itemize}
    \item CoMP-assisted OMA: The two cells are coordinating together to serve the weak UE, whereas, unlike the proposed CoMP-assisted NOMA and CoMP-assisted C-NOMA schemes, each cell adopts frequency division multiple access (FDMA) to serve its associated UEs. 
    \item C-NOMA: Each cell adopts C-NOMA to serve its associated UEs. Specifically, each strong UE has the ability to harvest the energy from the light intensity broadcast from the APs. Therefore, by harnessing the SIC capabilities, each strong UE forwards the decoded weak UE’s signal through RF D2D link using the harvested energy. In addition, the two cells exploit the same frequency bandwidth but without any coordination between them, unlike the proposed CoMP-assisted C-NOMA scheme.
    \item NOMA: Each cell adopts NOMA to serve its associated UEs. In addition, the two cells exploit the same frequency bandwidth but without any coordination between them, unlike the proposed CoMP-assisted NOMA scheme.
\end{itemize}
\begin{table}[t]
\caption{Simulation Parameters}
\centering
\renewcommand{\arraystretch}{0.5} 
\setlength{\tabcolsep}{0.18cm} 
\begin{tabular}{| c | c | c |}
  \hline 
  Parameter & Symbol & Value  \\
  \hline
  Distance between APs & $d_{\rm AP}$ & $4$m  \\
  \hline
  DC component & $I_{\rm DC}$ & $25$ dBm  \\ 
  \hline
  Current-to-power conversion factor & $\eta$ & $0.6$ W/A  \\ 
  \hline 
  Half-power semi-angle & $\phi_{1/2}$ & $45^\circ$  \\ 
  \hline 
  Optical concentrator refractive index & $n_c$ & 1 \\ 
  \hline
  Bandwidth of a VLC AP & $B_{\rm v}$ & $20$ MHz \\ 
  \hline
  UE's height & $H_{\rm u}$ & $0.9$m  \\
  \hline
  RF bandwidth & $B_{\rm r}$ & $16$ MHz \\
  \hline
  Fill factor & $f$ & $0.75$\\
  \hline
  RF fading parameter & $F$ & $1$  \\ 
  \hline
   Height of the APs & $H_{\rm AP}$ & 2.5$m$ \\
  \hline
   Modulation index & $\nu$ & $0.33$ \\ 
  \hline
  PD responsivity & $R_{\rm p}$ & $0.58$ A/W \\ 
  \hline 
  PD geometric area & $A_{\rm PD}$ & $1$ cm$^2$ \\ 
  \hline 
   Field of view of the PD & $\Psi$ & $60^\circ$ \\ 
  \hline
  Noise power spectral density & $N_0$ & $10^{-21}$ W/Hz \\ 
  \hline
   Number of discrete points & $K$ & 1000 \\
  \hline
   Thermal voltage & $V_t$ & $25$ mV \\
  \hline
   Dark saturation current of the PD & $I_0$ & $10^{-10}$ A\\
  \hline
   Path loss exponent & $\mu$ & $2$ \\ 
  \hline
\end{tabular} 
\label{T2}
\end{table}
\subsection{Sum Data Rate Performance}
\subsubsection{On the Optimality of the Proposed Solution Approaches}
\begin{figure}[t]
   	\centering     
	\includegraphics[width=1\linewidth]{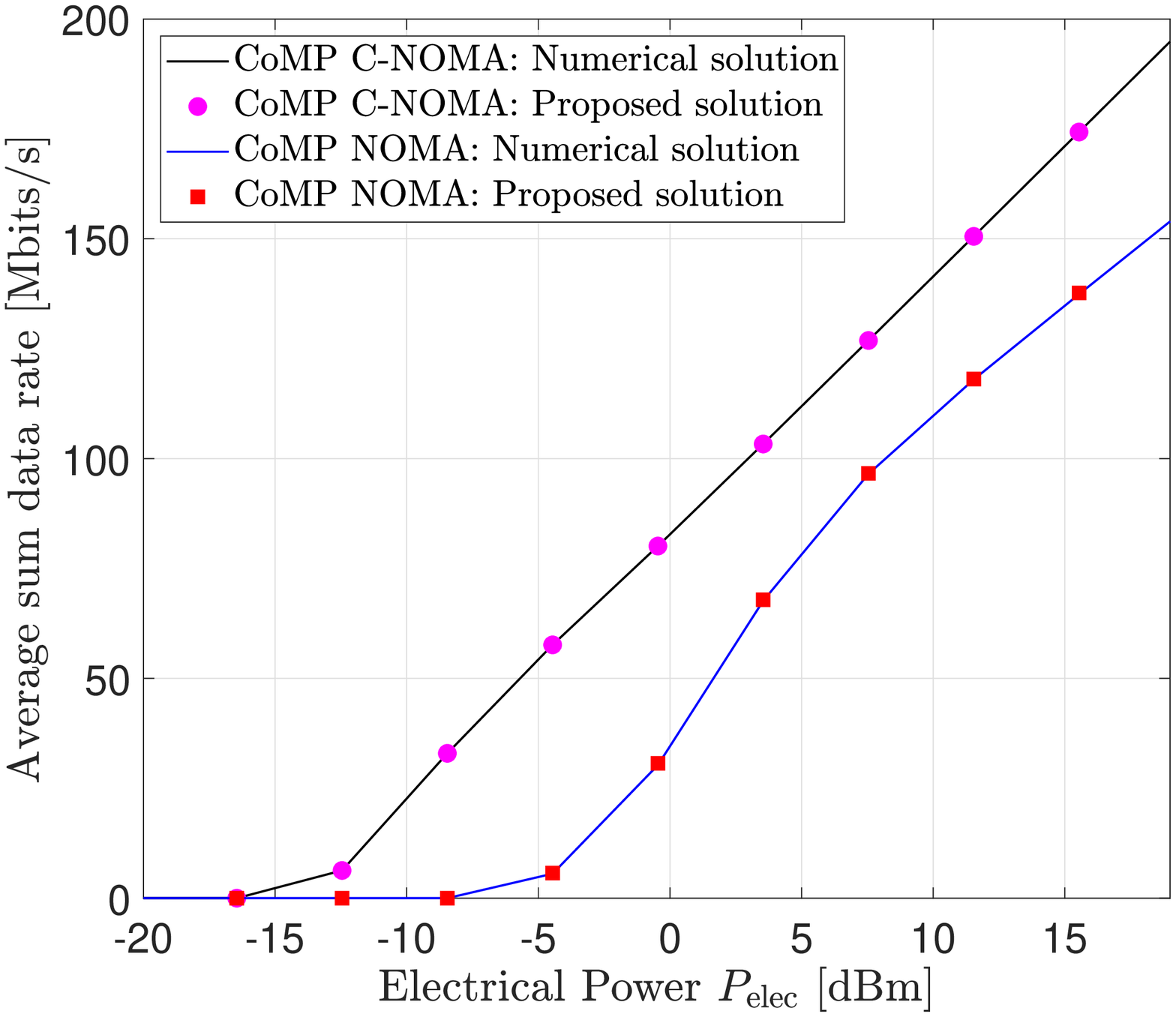}
    \captionof{figure}{Average sum data rate achieved by the numerical and the proposed solution approaches for the CoMP-assisted NOMA and the CoMP-assisted C-NOMA schemes versus the transmit electrical power $P_{\rm elec}$ at the APs.}
    \label{fig:analytical_numerical_sum_rate_power}
\end{figure}
\begin{figure}[t]
    \centering 
    \includegraphics[width=1\linewidth]{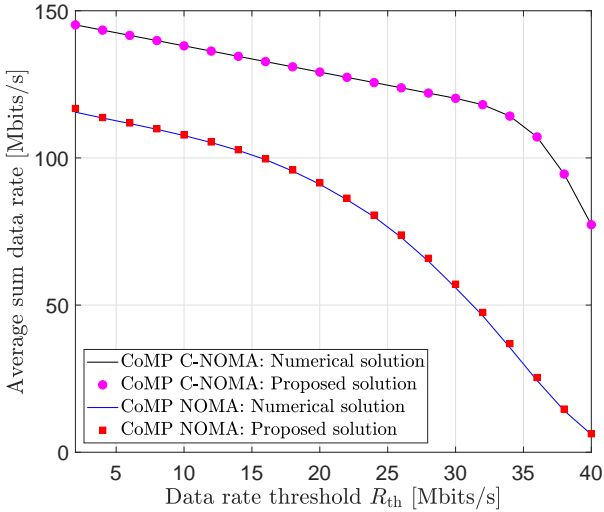}
    \captionof{figure}{Average sum data rate achieved by the numerical and the proposed solution approaches for the CoMP-assisted NOMA and the CoMP-assisted C-NOMA schemes versus the required data rate threshold $R_{\rm th}$ at the APs.}
    \label{fig:analytical_numerical_sum_rate_rth}
\end{figure}
\indent Fig.~\ref{fig:analytical_numerical_sum_rate_power} presents the average sum data rate, achieved by the proposed CoMP-assisted NOMA and CoMP-assisted C-NOMA schemes, versus the transmit electrical power $P_{\rm elec}$, when the required data rate threshold is $R_{\rm th} = 10$ [Mbit/Hz]. On the other hand, Fig.~\ref{fig:analytical_numerical_sum_rate_rth} presents the average sum data rate, achieved by the proposed CoMP-assisted NOMA and CoMP-assisted C-NOMA schemes, versus the required data rate threshold $R_{\rm th}$, when the transmit electrical power is $P_{\rm elec} = 9.5$ dBm. For the two proposed schemes, the results of both the numerical solutions and the proposed solutions of the power allocation coefficients are presented. The numerical solutions of problems $\mathcal{P}_1$ and $\mathcal{P}_2$ are obtained using an off-the-shelf optimization solver, whereas the analytical solutions are obtained through our proposed power allocation schemes.\footnote{\indent The adopted solver is $\tt{fmincon}$, which is a predefined $\tt{MATLAB}$ solver \cite{ebbesen2012generic}. In addition, $100$ distinct initial points were randomly generated within the feasibility region of the optimization variables in order to converge to the optimal solution. Specifically, through the use of the $\tt{fmincon}$ solver, each initial point will lead to a given local extremum. Then, once all local extrema are collected, a simple brute force search over the obtained extrema is applied to obtain the optimal solution. Although this heuristic approach suffers from its high complexity, it was demonstrated that it is very effective in finding the optimal solutions of non-convex problems as shown in \cite{huu2020low,elhattab2020power,elhattab2022ris,elhattab2022joint}. Nevertheless, it is important to mention that the optimal solutions of the considered problems can also be obtained using the technique proposed in \cite{papanikolaou2020optimal}, which has a complexity of $\mathcal{O}\left(\frac{1}{\epsilon^2}\right)$ when the target accuracy of the solution is less than $\epsilon$.} Fig.~\ref{fig:analytical_numerical_sum_rate_power} and \ref{fig:analytical_numerical_sum_rate_rth} show that the proposed solutions of the power allocation coefficients of the proposed CoMP-assisted NOMA and CoMP-assisted C-NOMA schemes match perfectly the numerical solutions, which demonstrates their optimality. Moreover, it demonstrates the superiority of the CoMP-assisted C-NOMA scheme over the CoMP-assisted NOMA scheme. This is basically due to the additional power at the strong UEs that is harvested from the APs and then used to relay the data of the weak UE. This harvested energy is resulting from the DC component received at the strong UEs, which is unexploited in the CoMP-assisted NOMA scheme.
\subsubsection{Effect of the Transmit Electrical Power $P_{\rm elec}$} 
\begin{figure}[t]
	\centering
	\begin{subfigure}[b]{0.5\columnwidth}
		\centering
		\includegraphics[width=\columnwidth,draft=false]{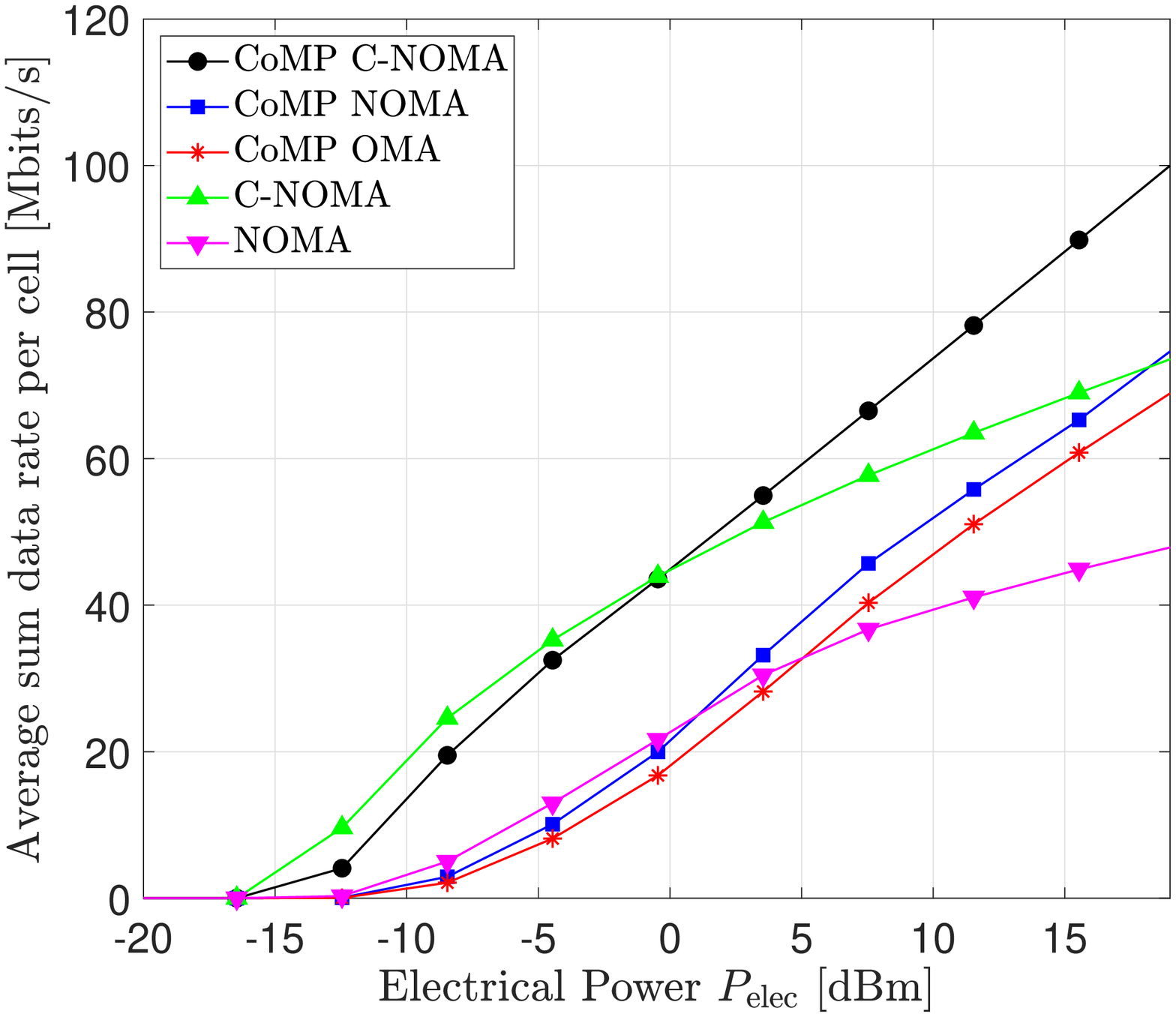}
		\caption{$d_{\rm AP} = 4$m, $\Phi_{1/2} = 45^\circ$.}
		\label{fig:I2d4phi45}
	\end{subfigure}%
	~
	\begin{subfigure}[b]{0.5\columnwidth}
		\centering
		\includegraphics[width=\columnwidth,draft=false]{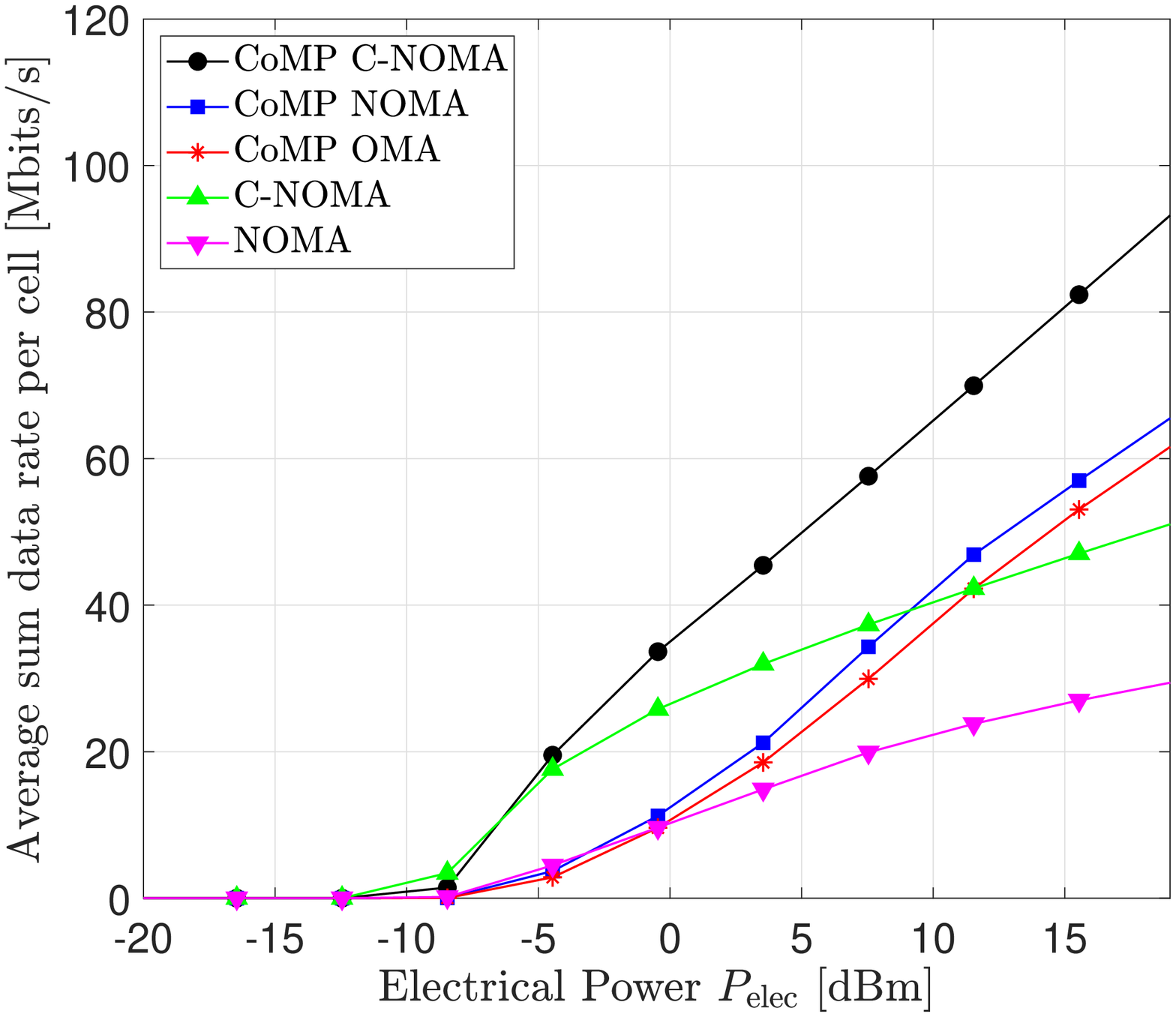}
		\caption{$d_{\rm AP} = 4$m, $\Phi_{1/2} = 60^\circ$.}
		\label{fig:I2d4phi60}
	\end{subfigure}\\
	\begin{subfigure}[b]{0.5\columnwidth}
		\centering
		\includegraphics[width=\columnwidth,draft=false]{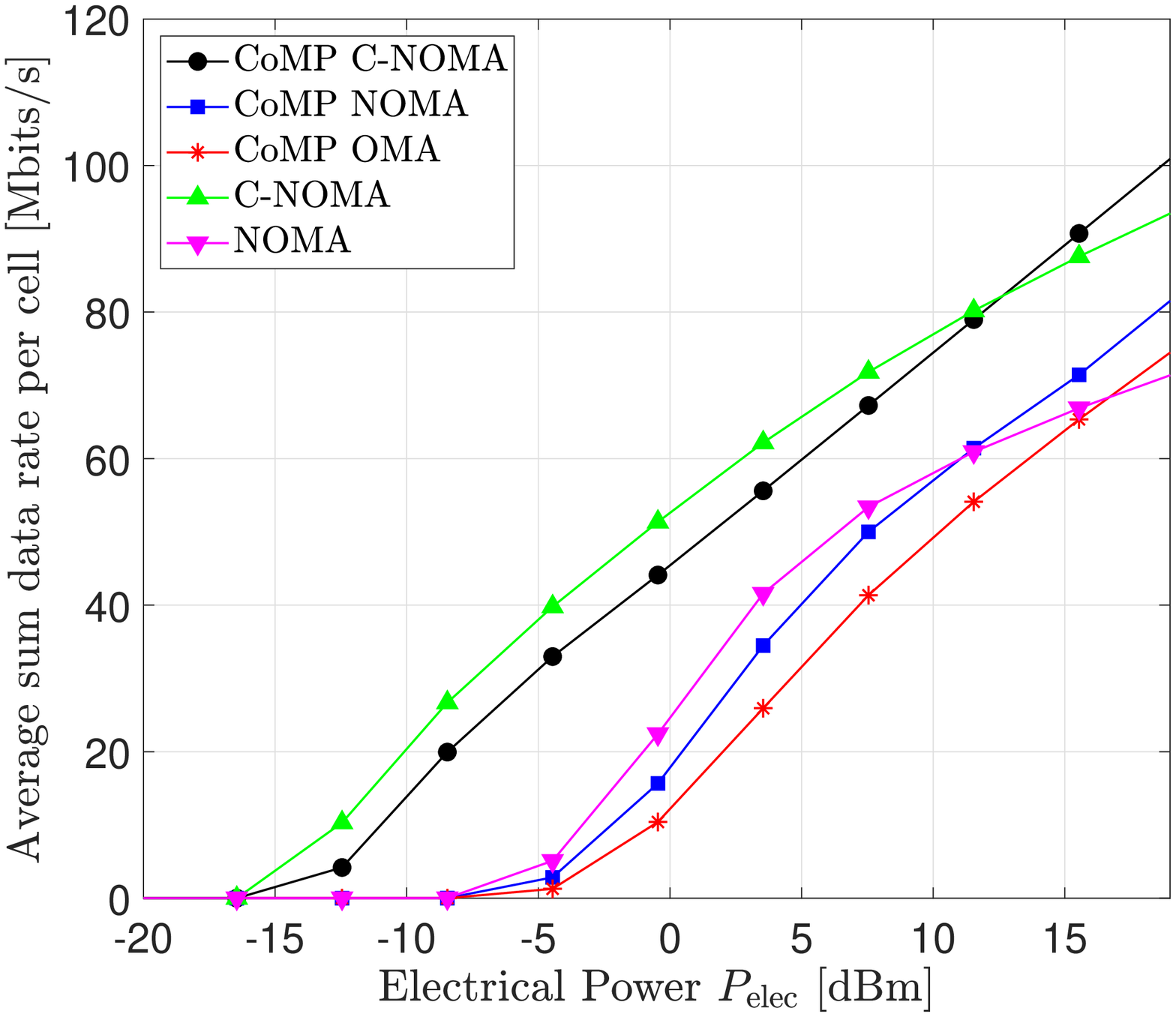}
		\caption{$d_{\rm AP} = 5$m, $\Phi_{1/2} = 45^\circ$.}
		\label{fig:I2d5phi45}
	\end{subfigure}%
	~
	\begin{subfigure}[b]{0.5\columnwidth}
		\centering
		\includegraphics[width=\columnwidth,draft=false]{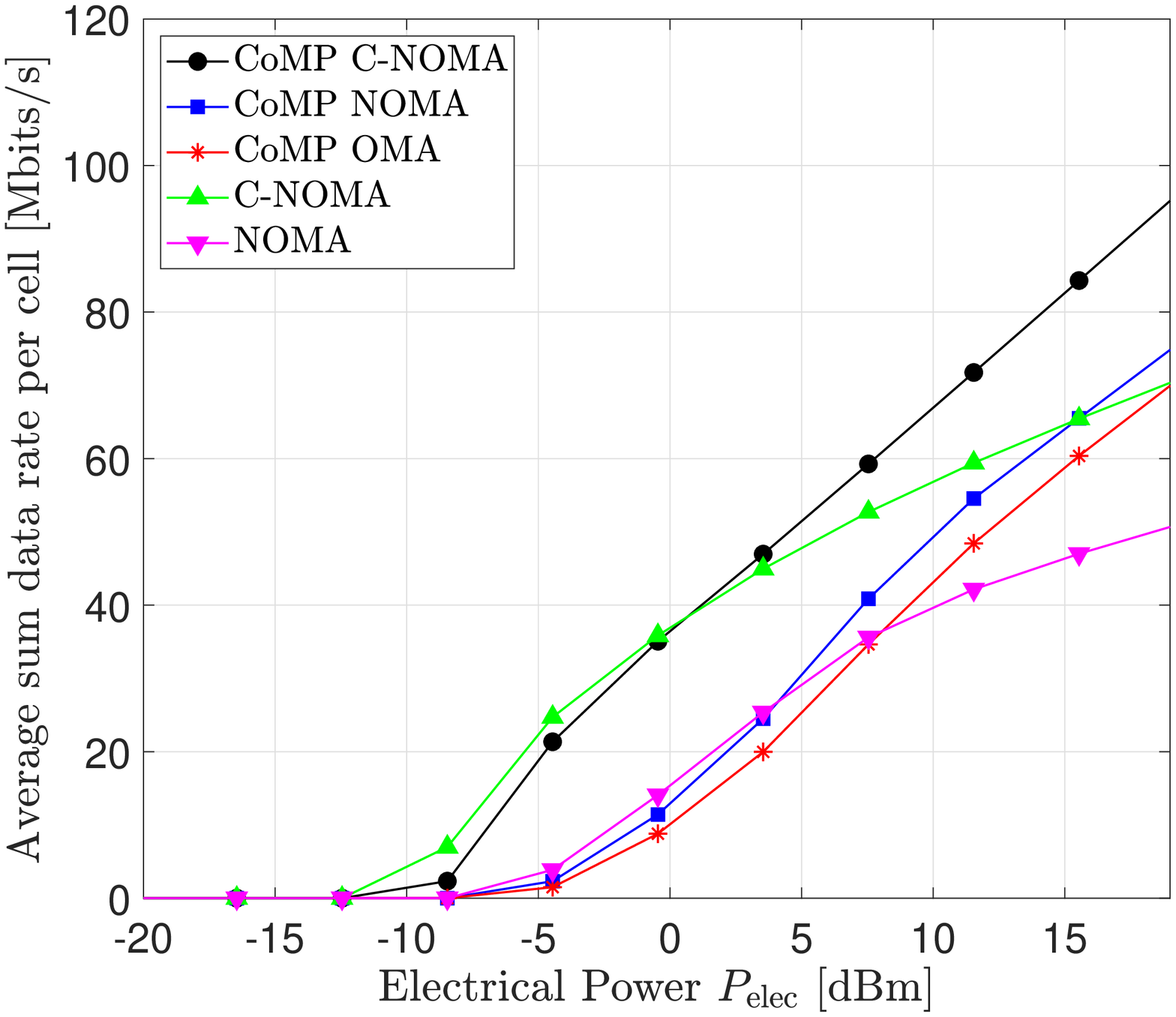}
		\caption{$d_{\rm AP} = 5$m, $\Phi_{1/2} = 60^\circ$.}
		\label{fig:I2d5phi60}
	\end{subfigure}
	\caption{Average achievable sum data rate per cell, achieved by the proposed CoMP-assisted NOMA and CoMP-assisted C-NOMA schemes, the CoMP-assisted OMA scheme, the C-NOMA scheme and the NOMA scheme, versus the transmit electrical power $P_{\rm elec}$.}
	\label{fig:sum_rate_power}
\end{figure}
\indent Fig.~\ref{fig:sum_rate_power} presents the average sum data rate per cell, achieved by the proposed CoMP-assisted NOMA and CoMP-assisted C-NOMA schemes, the CoMP-assisted OMA scheme, the C-NOMA scheme and the NOMA scheme, versus the transmit electrical power $P_{\rm elec}$, for different values of the distance between the APs $d_{\rm AP}$ and the height of the APs $H_{\rm AP}$. The required data rate threshold is $R_{\rm th} = 10$ [Mbit/Hz]. The results of Fig.~\ref{fig:sum_rate_power} can be divided into two regimes. namely, the low power regime and the high power regime. When the transmit electrical power is low, the C-NOMA scheme outperforms the CoMP-assisted C-NOMA scheme and the NOMA scheme outperforms the CoMP-assisted NOMA scheme, with C-NOMA being the best scheme that achieves the highest sum data rate per cell. However, when the transmit electrical power is high, the CoMP-assisted C-NOMA scheme outperforms the C-NOMA scheme and the CoMP-assisted NOMA scheme outperforms the NOMA scheme, with CoMP-assisted C-NOMA being the scheme that achieves the highest sum data rate per cell. In fact, when the transmit electrical power is low, the power of the interfering signals coming from the adjacent cell is low, and therefore, the use of the coordinated ZF precoding between the APs in the CoMP-assisted schemes is meaningless. However, when the transmit electrical power is high, the power of the interfering signals coming from the adjacent cell is high. Due to this, the performance of the NOMA and the C-NOMA schemes starts to stagnate when the transmit electrical power increases, whereas accounting to the use of the coordinated ZF precoding, the performance of the CoMP-assisted schemes increases continuously when the transmit electrical power increases.
\subsubsection{Effect of the Data Rate Threshold $R_{\rm th}$}
\begin{figure}[t]
		\centering
		\begin{subfigure}[b]{0.5\columnwidth}
			\centering
			\includegraphics[width=\columnwidth,draft=false]{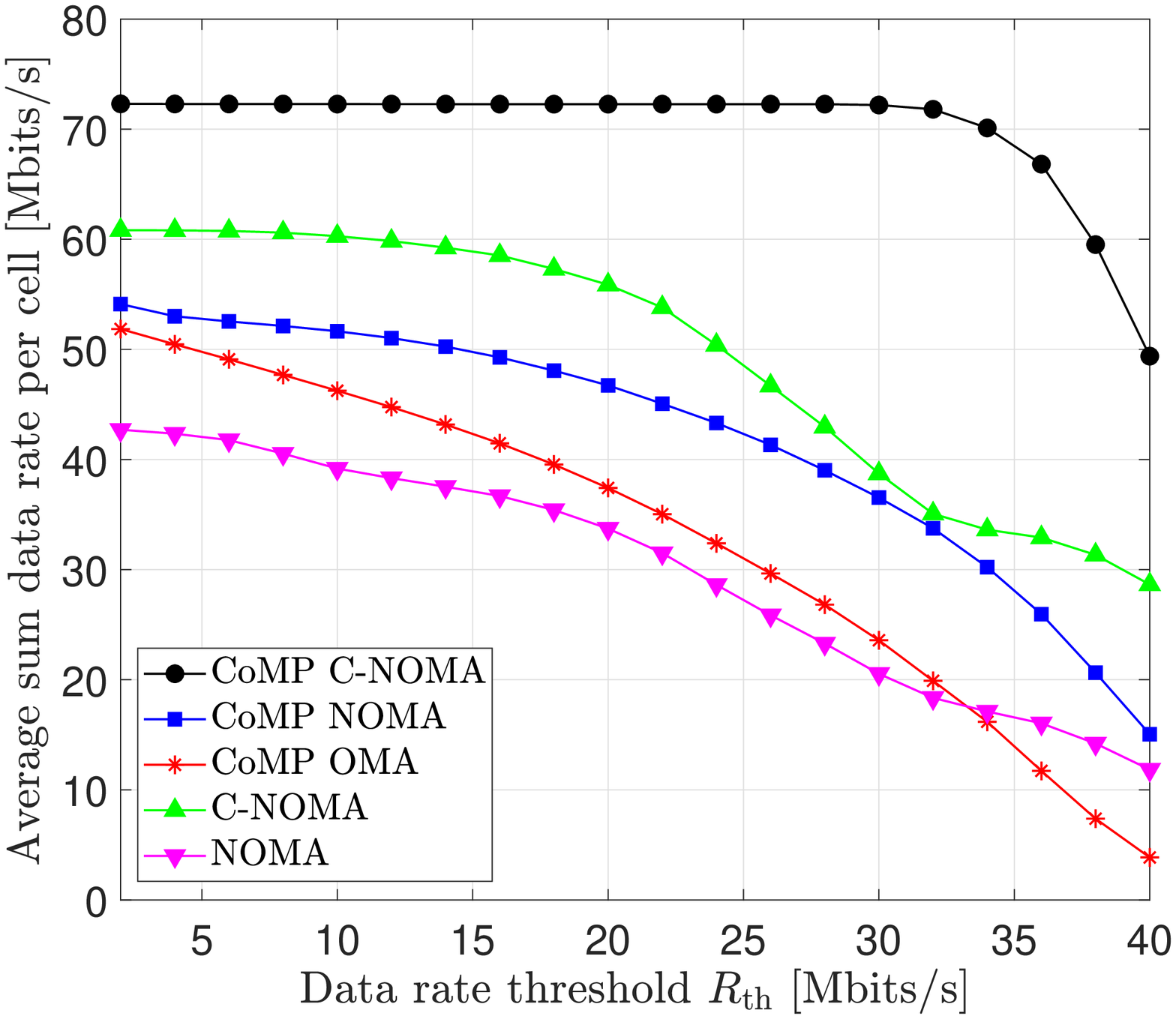}
			\caption{$d_{\rm AP} = 4$m, $H_{\rm AP} = 2.5$m.}
			\label{fig:d4phi45}
		\end{subfigure}%
		~
		\begin{subfigure}[b]{0.5\columnwidth}
			\centering
			\includegraphics[width=\columnwidth,draft=false]{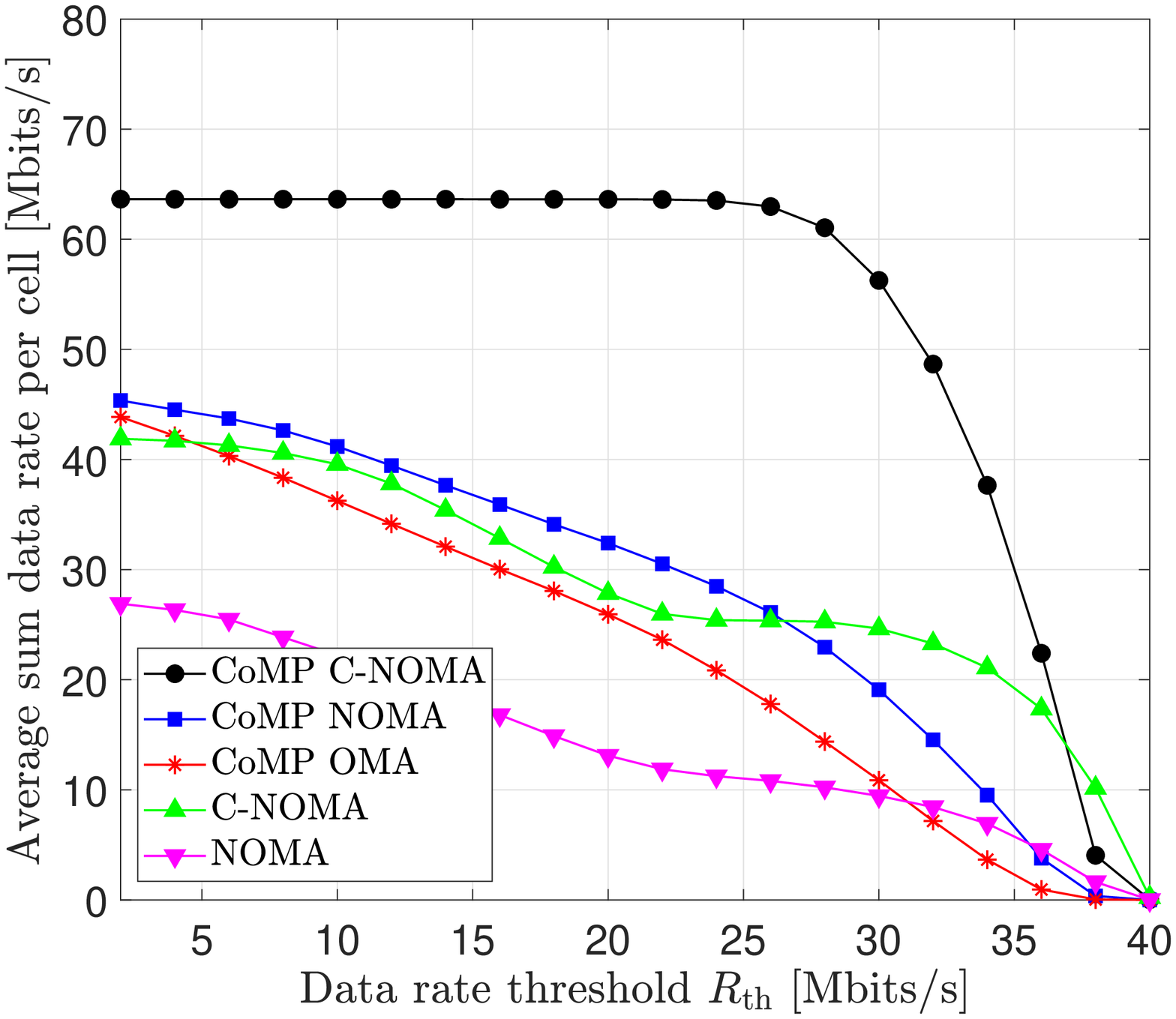}
			\caption{$d_{\rm AP} = 4$m, $H_{\rm AP} = 3$m.}
			\label{fig:d4phi60}
		\end{subfigure}\\
		\begin{subfigure}[b]{0.5\columnwidth}
			\centering
			\includegraphics[width=\columnwidth,draft=false]{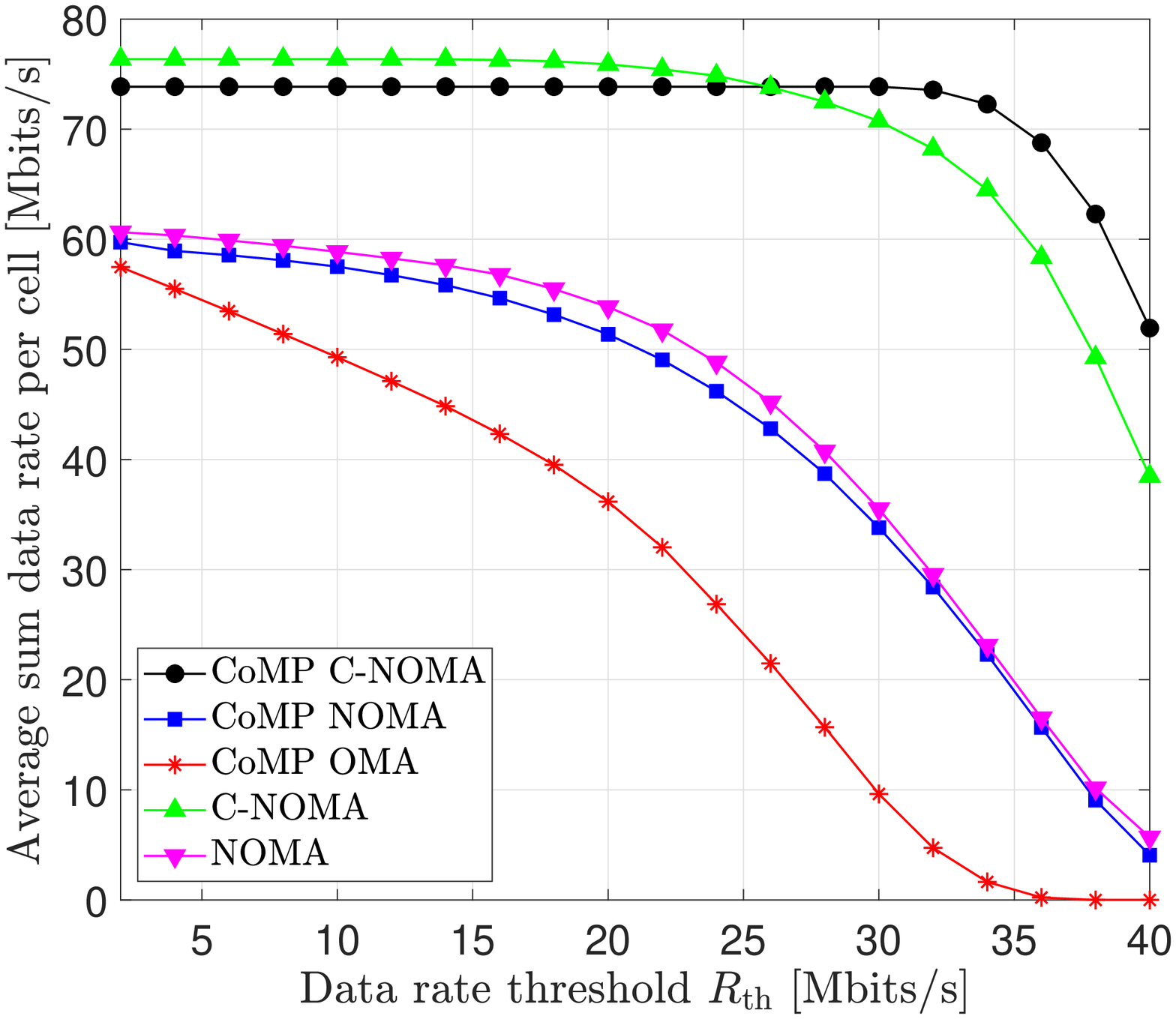}
			\caption{$d_{\rm AP} = 5$m, $H_{\rm AP} = 2.5$m.}
			\label{fig:d5phi45}
		\end{subfigure}%
		~
		\begin{subfigure}[b]{0.5\columnwidth}
			\centering
			\includegraphics[width=\columnwidth,draft=false]{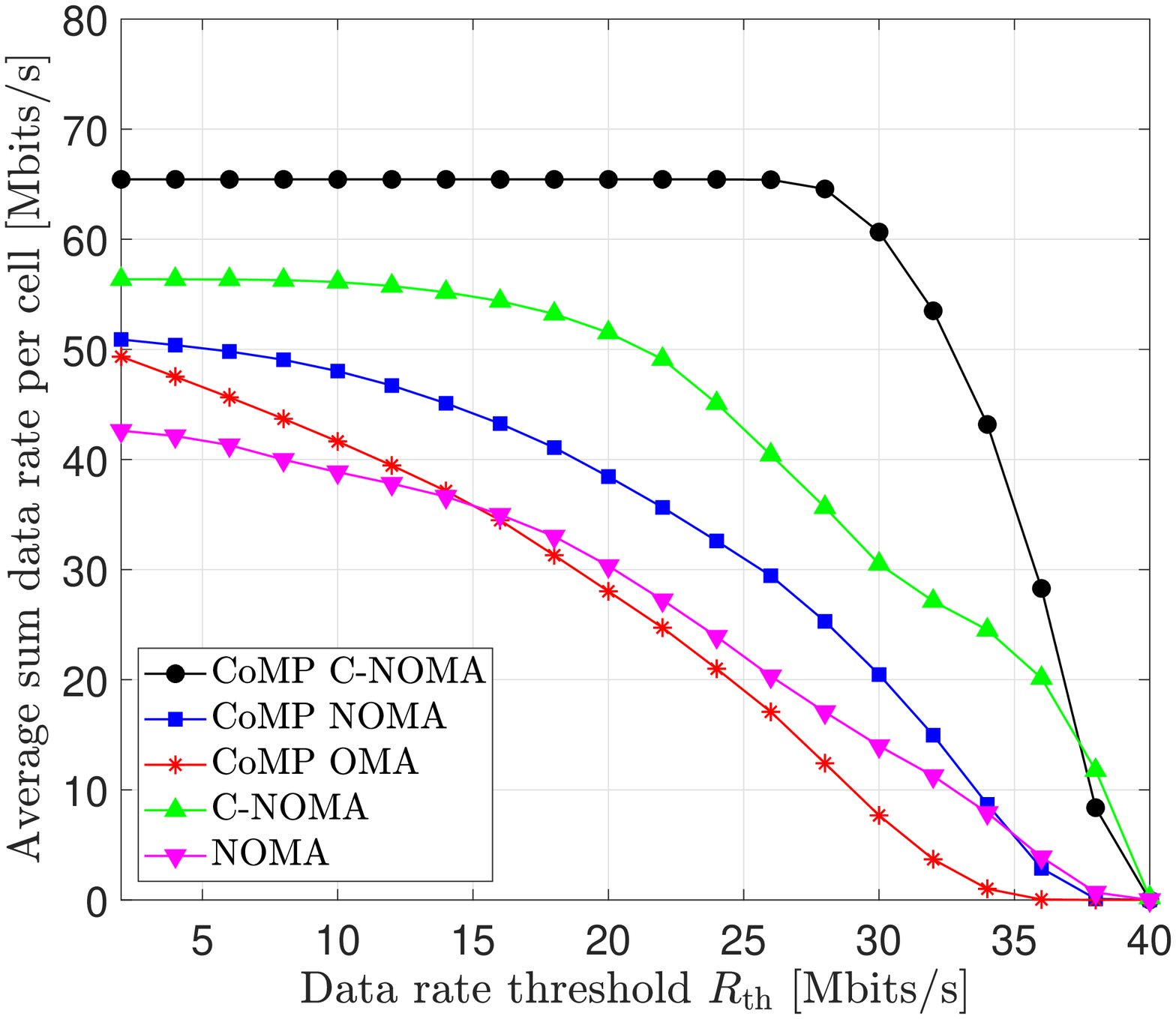}
			\caption{$d_{\rm AP} = 5$m, $H_{\rm AP} = 3$m.}
			\label{fig:d5phi60}
		\end{subfigure}
		\caption{Average achievable sum data rate per cell, achieved by the proposed CoMP-assisted NOMA and CoMP-assisted C-NOMA schemes, the CoMP-assisted OMA scheme, the C-NOMA scheme and the NOMA scheme, versus the required data rate threshold per UE $R_{\rm th}$.}
		\label{fig:sum_rate_threshold}
	\end{figure}
\indent Fig.~\ref{fig:sum_rate_threshold} presents the average sum data rate per cell, achieved by the proposed CoMP-assisted NOMA and CoMP-assisted C-NOMA schemes, the CoMP-assisted OMA scheme, the C-NOMA scheme and the NOMA scheme, versus the required data rate threshold per UE $R_{\rm th}$ for different values of the distance between the APs $d_{\rm AP}$ and the height of the APs $H_{\rm AP}$, where the transmit electrical power is $P_{\rm elec} = 9.5$ dBm. Fig.~\ref{fig:sum_rate_threshold} shows that, for all the considered cases of $\left(d_{\rm AP}, H_{\rm AP}\right)$, the proposed CoMP-assisted C-NOMA scheme outperforms all the considered baselines, which demonstrates the capability of integrating CoMP with C-NOMA in beating the ICI effects and increasing the network sum data rate simultaneously. On the other hand, we remark that the average sum data rate per cell achieved by all the considered schemes decrease when the rate threshold $R_{\rm th}$ increases. This is basically due to the fact that as $R_{\rm th}$ increases, the number of UEs that satisfy the target QoS decreases, and therefore, the number of UEs that can be served by the two APs decreases. The high slope of the drop in the achieved sum data rates by all schemes is mainly resulting from the $1/2$ pre-log penalties in the achievable data rate expression derived in the literature \cite{chaaban2016capacity}. \\
\indent In Fig.~\ref{fig:sum_rate_threshold}, one can see that the average achievable sum data rates achieved by all the considered schemes decrease as the distance between the APs $d_{\rm AP}$ decreases and/or the height of the APs $H_{\rm AP}$ increases. In fact, when $d_{\rm AP}$ decreases and/or $H_{\rm AP}$ increases, the coverage area of each AP increases, and hence, the area of the overlapping region between the two APs increases. Therefore, the ICI effects from one cell to the other increases, which explains the performance degradation for the C-NOMA and the NOMA schemes. On the other hand, despite the exploitation of the coordinated broadcasting technique between the two APs, the performance degradation of the CoMP-assisted schemes is resulting from the use of the ZF precoding and the peak-power constraint imposed on VLC systems. Specifically, the use of the JT and the ZF precoding eliminates the ICI effects at the strong and weak UEs. However, as the coverage area of each AP increases, the channel coefficients from the two APs to the strong UEs increases, and hence, the coefficients of the channel matrix $\mathbf{H}_{\rm a,b}$ in \eqref{eq:channel_matrix_cell-center} increases. Consequently, the multiplicative term $\frac{1}{||\mathbf{H}_{\rm a,b}^{\perp}||_{\infty}}$, which is required to make the ZF precoding matrix satisfy the peak-power constraint at the LEDs of the APs, decreases. Therefore, when the distance between the APs $d_{\rm AP}$ decreases and/or the height of the APs $H_{\rm AP}$ increases, the average received SNR at the UEs $\gamma_{Rx} = \frac{R_{\rm p}^2\eta^2P_{\rm elec}\sigma_{\rm s}^2}{\left|\left|\mathbf{H}_{\rm a,b}^{\perp}\right|\right|_{\infty}^2 \sigma_{\rm v}^2}$, which explains the performance degradation for the CoMP-assisted schemes, despite the use of the ICI mitigation techniques.
\subsubsection{Effects of the Area of the PD $A_{\rm PD}$} 
\indent The area of the PD $A_{\rm PD}$ at each UE is an important factor whose effects in the system performance should be analyzed. In this context, Fig.~\ref{fig:rate_vs_PD_area} presents the average sum data rate per cell, achieved by the proposed CoMP-assisted NOMA and CoMP-assisted C-NOMA schemes, the CoMP-assisted OMA scheme, the C-NOMA scheme and the NOMA scheme, versus the area of the PD of each UE $A_{\rm PD}$, when the transmit electrical power is $P_{\rm elec} = 9.5$ dBm and the required data rate threshold is $R_{\rm th} = 10$ [Mbit/Hz]. This figure demonstrates the superiority of the proposed CoMP-assisted schemes compared to the considered baselines and shows that the system performance increases when the area of the PD increases. This can be explained as follows. When the area of the PD increases, the received optical power from the APs at each UE increases, and hence, the power of the received signals increases, which increases the performance of the system. However, this is not the case for the C-NOMA and NOMA schemes, where at a certain point, increasing the area of the PD will increase the received optical power of both the useful and the interfering signals. Due to this, after a certain value of the PD's area, the sum data rate achieved by the C-NOMA and NOMA schemes starts stagnating.
\begin{figure}[t]
   	\centering     
	\includegraphics[width=1\linewidth]{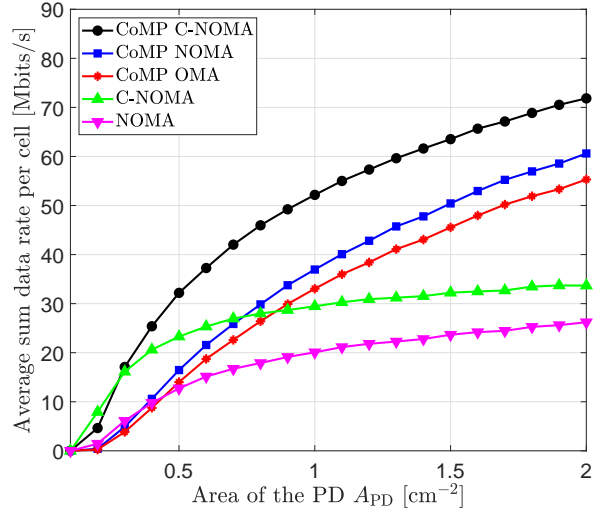}
    \captionof{figure}{Average achievable sum data rate per cell, achieved by the proposed CoMP-assisted NOMA and CoMP-assisted C-NOMA schemes, and the considered baselines, versus the area of the PD of each UE $A_{\rm PD}$.}
    \label{fig:rate_vs_PD_area}
\end{figure}
\begin{figure}[t]
     \centering 
    \includegraphics[width=1\linewidth]{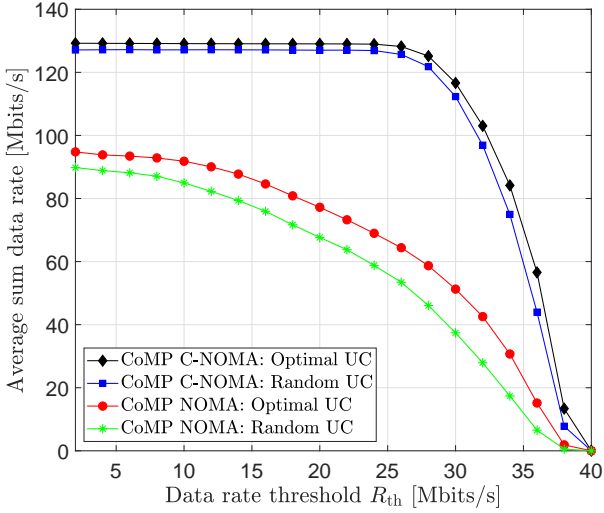}
    \captionof{figure}{Average achievable sum data rate achieved by the proposed CoMP-assisted NOMA and CoMP-assisted C-NOMA schemes versus the required data rate threshold per $R_{\rm th}$ for different UC techniques.}
    \label{fig:UC}
\end{figure}
\subsubsection{The Case of Multiple Users Per Cell} \quad \\
\indent In this part, we consider the case where six UEs are coexisting within the proposed network model, where the first two UEs are strong UEs located near the center of the first cell, and hence, associated to the first cell, the second two UEs are strong UEs located near the center of the second cell, and hence, associated to the second cell, and the remaining two UEs are weak UEs located near the edge of both cells, and hence, associated to both cells. Specifically, the locations of the strong UEs are generated randomly near the centers of the coverage areas of the APs, whereas the locations of the weak UEs are generated randomly near the edges of the cells, in which the orientation of each UE is randomly generated using the measurements-based orientation models proposed in \cite{soltani2018modeling,soltani2019bidirectional,arfaoui2020measurements}. As such, four UEs are associated to each AP. In this case, the UEs can be clustered into groups of three UEs. Each cluster consists of one strong UE associated to the first AP, one strong UE associated to the second AP, and one weak UE served jointly by the two APs. Since multiple clusters are constructed by the two APs, a promising technique to avoid the inter-cluster interference within each cell is to utilize a hybrid multiple access technique, in which the proposed CoMP-assisted NOMA and CoMP-assisted C-NOMA schemes are combined with the conventional FDMA scheme. In particular, the proposed CoMP-assisted NOMA and CoMP-assisted C-NOMA schemes are applied within each cluster of UEs, whereas different clusters are served through different sub-bandwidths. In this context, the UE clustering (UC) policy is another dimension to optimize along with the power allocation scheme. As such, two UC policies are considered for this purpose, namely, the optimal UC and the random UC, which are defined as follows.
\begin{itemize}
    \item Optimal UC: The optimal power allocation coefficients that maximize the sum data rate are obtained for all possible clusters using the proposed power allocation schemes. Then, the \textit{Hungarian} method is applied to obtain the optimal UC policy \cite{huu2020low,elhattab2022ris}.
    \item Random UC: Random disjoint clusters of UEs are generated and then the optimal power allocation coefficients that maximize the sum data rate are obtained for each cluster using the proposed power allocation schemes.
\end{itemize}
\indent Fig.~\ref{fig:UC} presents the average overall sum data rate of the VLC network achieved by the proposed CoMP-assisted NOMA and CoMP-assisted C-NOMA schemes versus the required data rate threshold per UE $R_{\rm th}$ for the considered UC techniques. This figure shows that the optimal UC policy achieves better performance than the random clustering.
\subsection{Minimum Data Rate Performance}
\begin{figure}
\centering     
\includegraphics[width=1\linewidth]{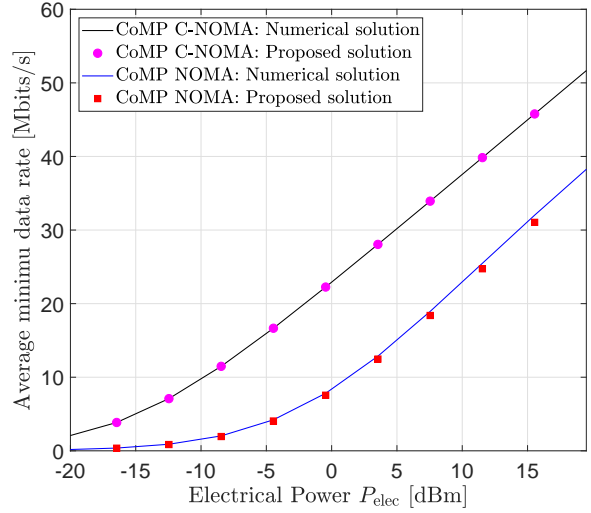}
\caption{Average minimum data rate achieved by the numerical and the proposed solution approaches for the CoMP-assisted NOMA and the CoMP-assisted C-NOMA schemes versus the transmit electrical power $P_{\rm elec}$ at the APs.}
\label{fig:analytical_numerical_min_rate}
\end{figure}
\subsubsection{On the Optimality of the Proposed Solution Approaches} 
\indent Fig.~\ref{fig:analytical_numerical_min_rate} presents the average minimum data rate, achieved by the proposed CoMP-assisted NOMA and the CoMP-assisted C-NOMA schemes, versus the transmit electrical power $P_{\rm elec}$. For the two proposed schemes, the results of both the numerical solutions and the proposed solutions of the power allocation coefficients are presented. The numerical solutions are obtained by solving problems $\mathcal{P}_3$ and $\mathcal{P}_4$ using an off the-shelf optimization solver, whereas the analytical solutions are obtained through our proposed power allocation schemes. Fig.~\ref{fig:analytical_numerical_min_rate} shows that the proposed solutions of the power allocation coefficients of the proposed CoMP-assisted NOMA and CoMP-assisted C-NOMA schemes match perfectly the numerical solutions, which demonstrates their optimality. Moreover, and as was shown by Fig.~\ref{fig:analytical_numerical_sum_rate_power}, Fig.~\ref{fig:analytical_numerical_min_rate} demonstrates the superiority of the CoMP-assisted C-NOMA scheme over the CoMP-assisted NOMA scheme due to the additional power at the strong UEs that is harvested from the APs and then used to relay the data of the weak UE. 
\subsubsection{Effect of the Transmit Electrical Power $P_{\rm elec}$} 
\begin{figure}[t]
	\centering
	\begin{subfigure}[b]{0.5\columnwidth}
		\centering
		\includegraphics[width=\columnwidth,draft=false]{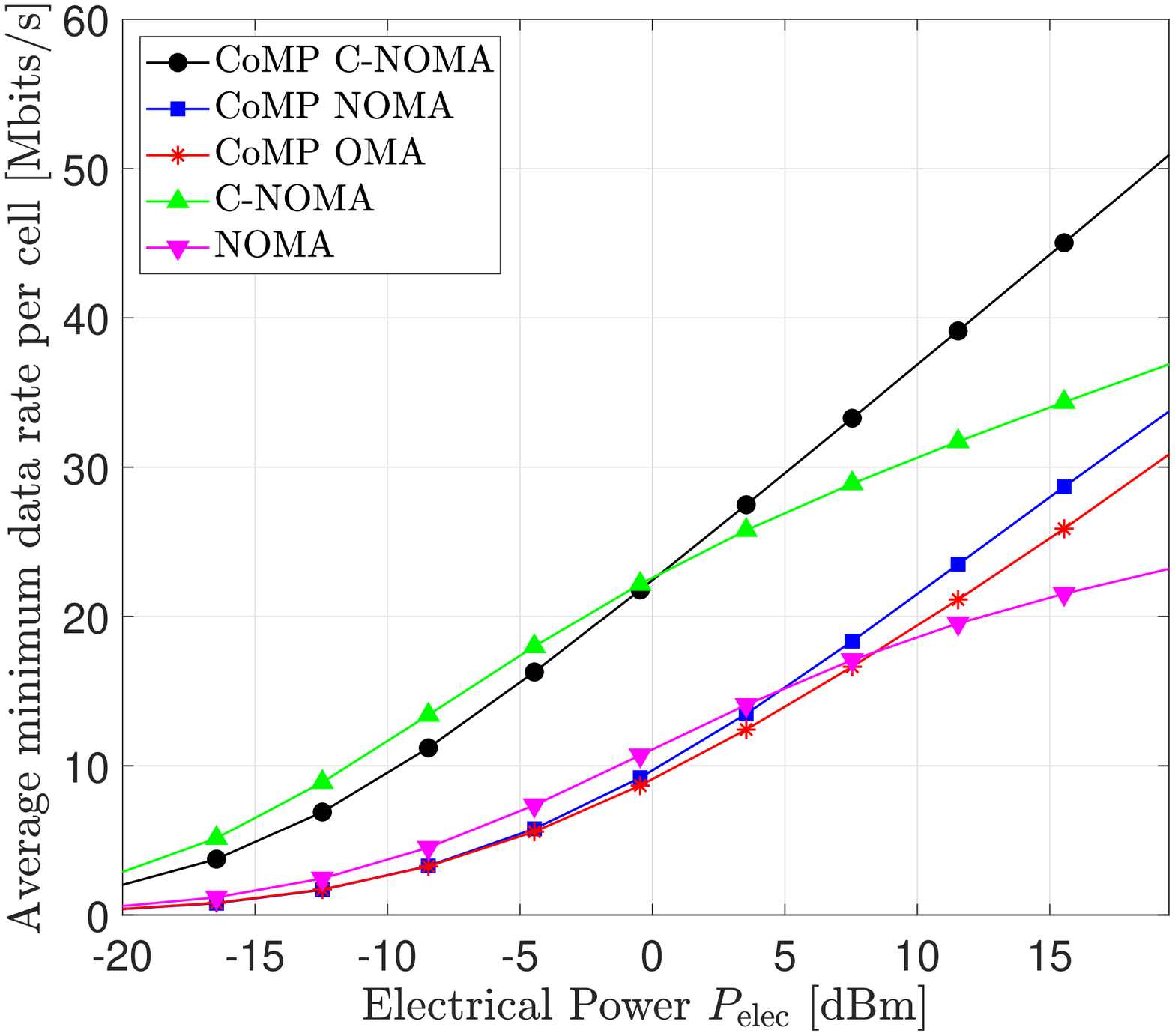}
		\caption{$d_{\rm AP} = 4$m, $H_{\rm AP} = 2.5$m.}
		\label{fig:mI2d4phi45}
	\end{subfigure}%
	~
	\begin{subfigure}[b]{0.5\columnwidth}
		\centering
		\includegraphics[width=\columnwidth,draft=false]{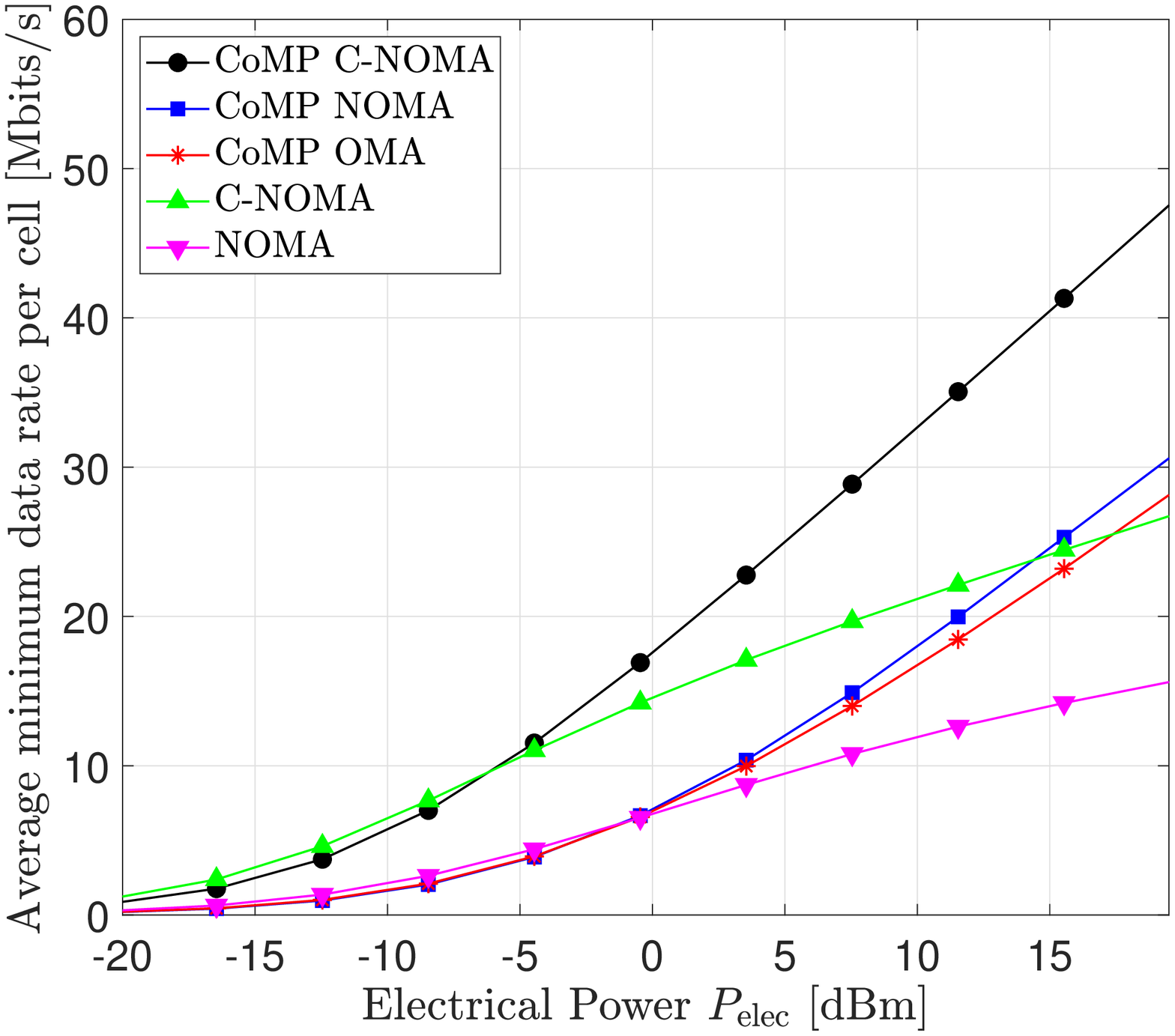}
		\caption{$d_{\rm AP} = 4$m, $H_{\rm AP} = 3$m.}
		\label{fig:mI2d4phi60}
	\end{subfigure}\\
	\begin{subfigure}[b]{0.5\columnwidth}
		\centering
		\includegraphics[width=\columnwidth,draft=false]{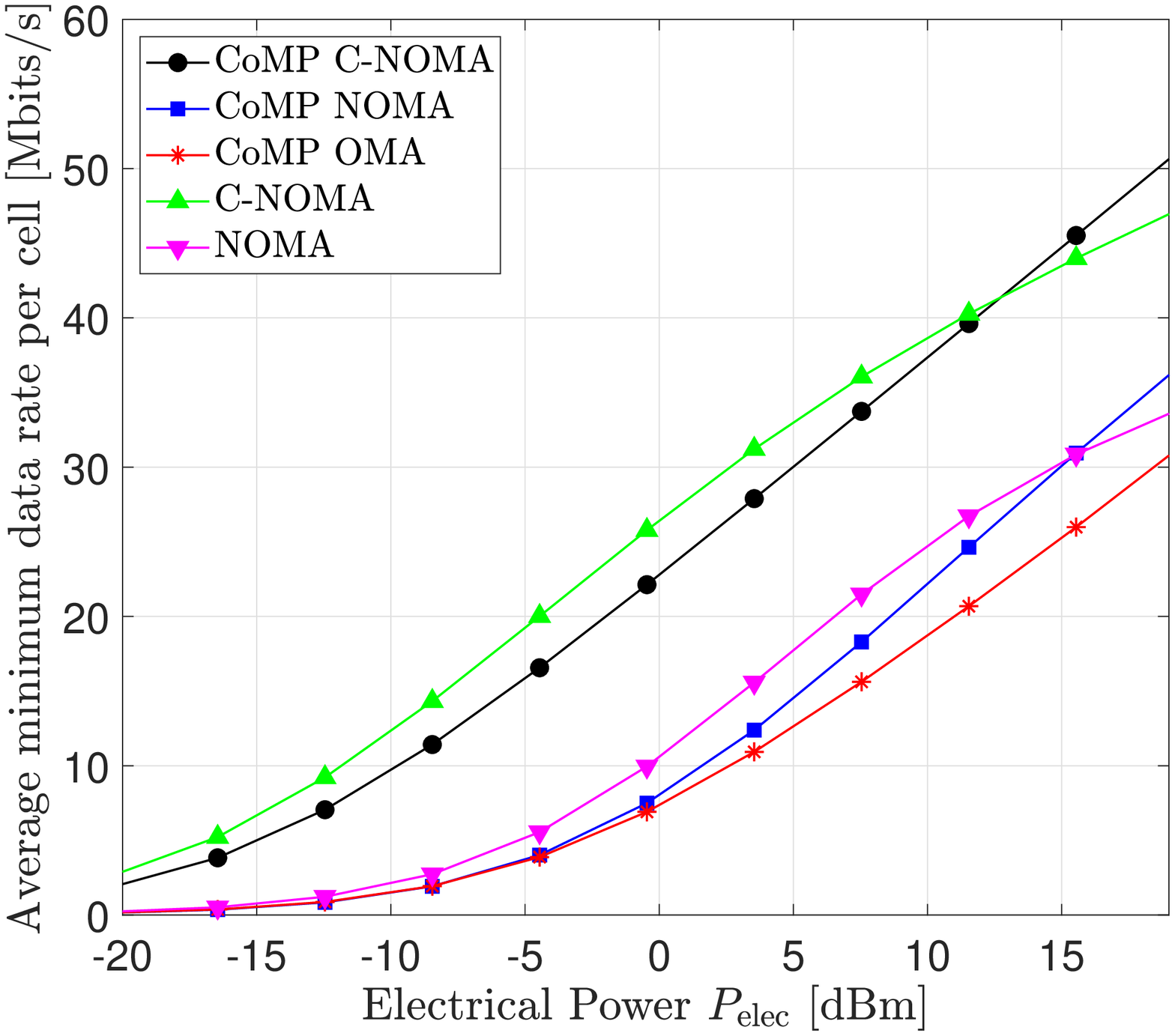}
		\caption{$d_{\rm AP} = 5$m, $H_{\rm AP} = 2.5$m.}
		\label{fig:mI2d5phi45}
	\end{subfigure}%
	~
	\begin{subfigure}[b]{0.5\columnwidth}
		\centering
		\includegraphics[width=\columnwidth,draft=false]{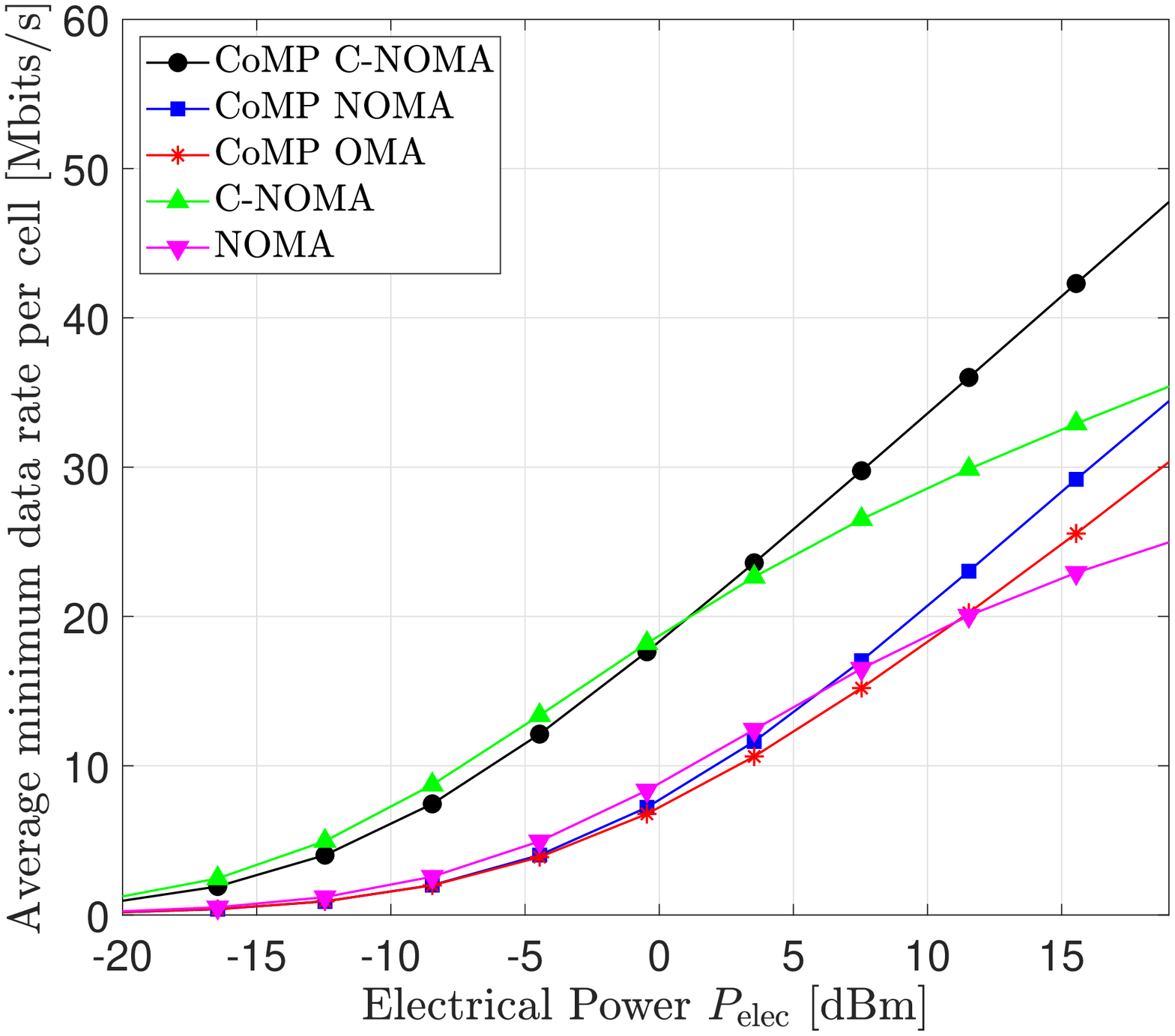}
		\caption{$d_{\rm AP} = 5$m, $H_{\rm AP} = 3$m.}
		\label{fig:mI2d5phi60}
	\end{subfigure}
	\caption{Average minimum data rate per cell, achieved by the proposed CoMP-assisted NOMA and CoMP-assisted C-NOMA schemes, the CoMP-assisted OMA scheme, the C-NOMA scheme and the NOMA scheme, versus the transmit electrical power $P_{\rm elec}$.}
	\label{fig:min_rate_power}
\end{figure}
\indent Fig.~\ref{fig:min_rate_power} presents the average minimum data rate per cell, achieved by the proposed CoMP-assisted NOMA and CoMP-assisted C-NOMA schemes, the CoMP-assisted OMA scheme, the C-NOMA scheme and the NOMA scheme, versus the transmit electrical power $P_{\rm elec}$, for different values of the distance between the APs $d_{\rm AP}$ and the height of the APs $H_{\rm AP}$. Basically, the same observations and interpretation that were remarked for the sum data rate performance in Fig.~\ref{fig:sum_rate_power} are also valid for the minimum data rate performance.
\section{Conclusion}
This paper studies the performance of CoMP transmission in downlink multi-cell NOMA/C-NOMA VLC systems. For a system consisting of two adjacent attocells serving three coexisting users, optimal and low-complexity power control schemes that maximizes the network sum data rate, while guaranteeing target QoS at the end users, and the minimum data rate within the network are derived. In the simulation results, the optimality of the derived power control schemes is verified and the performance of the proposed CoMP-assisted NOMA and CoMP-assisted C-NOMA schemes are compared with those of the CoMP-assisted FDMA scheme, the C-NOMA scheme and the NOMA scheme, where the superiority of the proposed schemes is demonstrated. The extension of the proposed CoMP-assisted NOMA and CoMP-assisted C-NOMA schemes to multiple coordinating attocells, i.e., a number of APs higher than two, can be considered as a potential future research direction. In such a case, although the dynamic power control and the interference management become more challenging, it is expected that a performance enhancement of VLC cellular systems can be achieved. 
\section*{Acknowledgment}
This work was supported in part by the Qatar National Research Fund (a member of Qatar Foundation) under Grant AICC03-0324-200005, in part by Concordia University, in part by the les Fonds de Recherche du Québec—Nature et Technologies (FRQNT), and in part by the Natural Sciences and Engineering Research Council of Canada (NSERC).
\appendices 
\section{The Peak Power Constraint}
\label{Appendix:A}
Let us assume that the precoding matrix $\mathbf{W}$ satisfy the inequality $||\mathbf{W}||_{\infty} \leq 1$ and that the total electrical power $P_{\rm elec}$ satisfy the inequality $P_{\rm elec} \leq \frac{\left(\nu I_{\rm DC}\right)^2}{2}$. The infinity norm of the vector of precoded messages $\mathbf{W} \mathbf{s}$ satisfy $||\mathbf{W} \mathbf{s}||_{\infty} \leq ||\mathbf{W}||_{\infty} ||\mathbf{s}||_{\infty} \leq ||\mathbf{s}||_{\infty}$. On the other hand, we have $|s_1| = \left|\sqrt{(1-\alpha_1) P_{\rm elec}} s_{\rm a} +  \sqrt{\alpha_1P_{\rm elec}} s_{\rm w}\right| \leq \sqrt{(1-\alpha_1) P_{\rm elec}} |s_{\rm a}| +  \sqrt{\alpha_1P_{\rm elec}} |s_{\rm w}|$. Now, recall that, for all $k \in \left\{\rm a,b,w \right\}$, the message $s_{k} \in \left[ -1, 1\right]$. Furthermore, note that, for all $\alpha_1 \in [0,1]$, we have $\sqrt{1-\alpha_1}  +  \sqrt{\alpha_1} \leq \sqrt{2}$ and that the equality holds if and only if $\alpha_1 = \frac{1}{2}$. Hence, $|s_1| \leq \sqrt{2P_{\rm elec}}$. Or, since $P_{\rm elec} \leq \frac{\left(\nu I_{\rm DC}\right)^2}{2}$, we obtain $|s_1| \leq \nu I_{\rm DC}$. Similarly, one can find that $|s_2| \leq \nu I_{\rm DC}$. Based on this, we have $||\mathbf{s}||_{\infty} \leq \nu I_{\rm DC}$, and therefore, $||\mathbf{W} \mathbf{s}||_{\infty} \leq \nu I_{\rm DC}$, which completes the proof.
\section{VLC Channel Model}
\label{Appendix:B}
\indent The downlink channel gain $h$ between an AP and a VLC receiver is expressed as $
h = h^\mathrm{LOS}+h^\mathrm{NLOS}$, such that $h^\mathrm{LOS}$ and $h^\mathrm{NLOS}$ denote the line of sight (LOS and the non line of sight (NLOS) channel gains, respectively. The LOS channel gain is expressed as \cite{arfaoui2021invoking,kahn1997wireless}
\begin{equation}
\label{eq:LOS}
     h^\mathrm{LOS}=\frac{(m+1)A_{\rm PD}}{2 \pi d^2} \cos^m(\phi) \cos(\psi) \mathrm{rect}\left(\frac{\psi}{\Psi}\right),
\end{equation}
where $m=-1/\log_2(\cos(\Phi_{1/2}))$ is the Lambertian emission order of the LEDs, $A_{\rm PD}$ is the area of the PD, $\phi$ is the angle of radiance, $\psi$ is the angle of incidence, $d$ is the distance between the AP and the UE. and $\Psi \in \left[0, \pi/2\right]$ is the field of view (FOV) of the PD of the UE, which defines the acceptance angle of PD. In other words, the FOV is defined such that the UE only detects light whose angle of incidence with respect to the PD's normal vector is less than the FOV \cite{barry1993simulation,lee2011indoor}.\\
\indent Concerning the NLOS components of the channel gain, they can be calculated based on the method described in \cite{NLOSSchulze}. Using the frequency domain instead of the time domain analysis, one is able to consider an infinite number of reflections to have an accurate value of the diffuse link. The environment is segmented into a number of surface elements which reflect the light beams. These surface elements are modeled as Lambertian radiators described by \eqref{eq:LOS} with $m=1$ and FOV of $90^\circ$. Assuming that the entire room can be decomposed into $M$ surface elements, the NLOS channel gain $h^\mathrm{NLOS}$, including an infinite number of reflections between the AP and the VLC receiver, can be expressed as $h^\mathrm{NLOS}=\mathbf{r}^\mathrm{T}\mathbf{G}_\zeta(\mathbf{I}_M-\mathbf{E}\mathbf{G}_\zeta)^{-1}\mathbf{t}$, where the vectors $\mathbf{t}$ and $\mathbf{r}$ respectively represent the LOS link between the AP and all the surface elements of the room and from all the surface elements of the room to the VLC receiver \cite{NLOSSchulze}. The matrix $\mathbf{G}_\zeta={\rm{diag}}(\zeta_1,...,\zeta_M)$ is the reflectivity matrix of all $M$ reflectors; $\mathbf{E}$ is the LOS transfer function of size $M\times M$ for the links between all surface elements, and $\mathbf{I}_M$ is the unity matrix of order $M$. The elements of $\mathbf{E}$, $\mathbf{r}$ and $\mathbf{t}$ are found according to \eqref{eq:LOS}.
\section{Proof of Theorem 1}
\label{Appendix: Theorem 1}
Considering the constraints \eqref{Const:C3} and \eqref{Const:C5}, and by substituting $R_{\rm a}$ and $R_{\rm b}$ with their expressions into \eqref{Const:C3} and \eqref{Const:C5}, and then solving the resulting inequality, we obtain $\alpha_1 \leq \alpha_{\max}$ and $\alpha_2 \leq \alpha_{\max}$. Now, considering the constraints \eqref{Const:C4} and \eqref{Const:C6}, and by substituting $R_{\rm a \rightarrow w}$ and $R_{\rm b \rightarrow w}$ with their expressions into \eqref{Const:C4} and \eqref{Const:C6}, and then solving the resulting inequality, we obtain $\alpha_{\min} \leq \alpha_1$ and $\alpha_{\min} \leq \alpha_2$. Based on this, constraints \eqref{Const:C1}-\eqref{Const:C6} are satisfied if and only if $\alpha_{\min} \leq \alpha_{\max}$, which constitutes the first feasibility condition of problem $\mathcal{P}_1$. Finally, we focus on constraint \eqref{Const:C7}. By substituting $R_{\rm w \rightarrow w}^{\rm VL}$ with its expression into constraint \eqref{Const:C7}, and solving the resulting inequality, we obtain the inequality $g\left(\alpha_1,\alpha_2\right) \geq 0$. Obviously, the last inequality is feasible if and only if it is satisfied by the highest values of $\alpha_1$ and $\alpha_2$, i.e., $\alpha_1 = \alpha_{\max}$ and $\alpha_2 = \alpha_{\max}$, which constitutes the second and last feasibility condition of problem $\mathcal{P}_1$.
\section{Line Search Method for the Sum data rate Maximization}
\label{Appendix: line search sum rate}
We assume that the power allocation coefficient $\alpha_1$ is fixed and we consider the change of variable $\beta = \sqrt{\alpha_2}$. Hence, solving the inequality $g\left(\alpha_1,\alpha_2\right) \geq 0$ is equivalent to solving the inequality $c_1 \beta^2 + c_2 \beta + c_3 \geq 0$ over the segment $\left[\sqrt{\alpha_{\min}}, \sqrt{\alpha_{\max}} \right]$, where $c_1 = c\left(1+t_{\rm v} \right)\Tilde{h}_{\rm 2,w}^2$, $c_2 = c\Tilde{h}_{\rm 1,w}\Tilde{h}_{\rm 2,w} \sqrt{\alpha_1}$, and $c_3 = c\left(1+t_{\rm v} \right)\Tilde{h}_{\rm 1,w}^2\alpha_1 - t \left(c\Tilde{h}_{\rm 1,w}^2 + c\Tilde{h}_{\rm 2,w}^2 + \frac{1}{\gamma_{\rm RX}} \right)$. We compute the discriminant $\Delta = c_2^2 - 4c_1c_2$. If $\Delta \leq 0$, and since $c_1 \geq 0$, then the lowest value of $\alpha_2$ that satisfies $g\left(\alpha_1,\alpha_2\right) \geq 0$ is $\alpha_{\min}$. Otherwise, we compute the ordered roots $\beta_1 = \min \left( \frac{-c_2-\sqrt{\Delta}}{2 c_1}, \frac{-c_2+\sqrt{\Delta}}{2 c_1}\right)$ and $\beta_2 = \max \left( \frac{-c_2-\sqrt{\Delta}}{2 c_1}, \frac{-c_2+\sqrt{\Delta}}{2 c_1}\right)$. In this case, the lowest value of $\alpha_2$ that satisfies $g\left(\alpha_1,\alpha_2\right) \geq 0$ is $\alpha_{\min}$, except the two following two cases. The first case is when $\beta_1 \leq \sqrt{\alpha_{\min}}$ and $\sqrt{\alpha_{\max}} \leq \beta_2$. In this case, there is no feasible solution for $\alpha_2$ and the associated achievable sum data rate is zero. The second case if when $\beta_1 \leq \sqrt{\alpha_{\min}}$ and $\sqrt{\alpha_{\min}} \leq \beta_2 \leq \sqrt{\alpha_{\max}}$. In this case, the lowest value of $\alpha_2$ that satisfies $g\left(\alpha_1,\alpha_2\right) \geq 0$ is $\beta_2^2$.
\bibliographystyle{IEEEtran}
\bibliography{main.bib}
\end{document}